\documentclass[11pt,draftclsnofoot,onecolumn]{IEEEtran}
\usepackage{amsmath,epsfig} 

\usepackage{amsmath}
\usepackage{amsthm}
\usepackage{amsfonts}
\usepackage{amssymb}
\usepackage{cite}
\usepackage{graphicx}
\usepackage{cases}
\usepackage{cancel}
\usepackage{subfigure}
\usepackage{float}
\usepackage{color}

\newcommand{\defn}{\,{\stackrel{\triangle}{=}}\,}

\def\m(#1){\mathcal{#1}}

\def\row(#1,#2){#1^{\textrm{row}}_#2}
\def\col(#1,#2){#1^{\textrm{col}}_#2}

\newcommand{\sinc}{\operatorname{sinc} }
\newcommand{\rect}{\operatorname{rect} }

\newcommand{\Kset}{\mathcal{K}}
\newcommand{\tabW}{p{2.4cm}}

\def\colorname(#1){#1}
\newtheorem{theorem}{\colorname(Theorem)}

\newtheorem{corollary}{\colorname(Corollary)}

\newcommand{\RonenComment}[1]{#1} 

\title{Innovation Rate Sampling of Pulse Streams with Application to Ultrasound Imaging}
\author{Ronen Tur, Yonina C.~Eldar,~\IEEEmembership{Senior~Member,~IEEE} and Zvi Friedman\thanks{Copyright (c) 2010 IEEE. Personal use of this material is permitted. However, permission to use this material for any other purposes must be obtained from the IEEE by sending a request to pubs-permissions@ieee.org.

Department of Electrical
Engineering, Technion---Israel Institute of Technology, Haifa 32000, Israel. Phone: +972-4-8293256, fax: +972-4-8295757,
E-mail: \{ronentur@techunix,yonina@ee\}.technion.ac.il, zvi.friedman@med.ge.com. Y. Eldar is currently a Visiting Professor at Stanford, CA. This work was supported in part by a
Magneton grant from the Israel Ministry of Industry and Trade.}}

\begin{document}
\bibliographystyle{IEEEtran}
\maketitle
\begin{abstract}
Signals comprised of a stream of short pulses appear in many
applications including bio-imaging and radar. The recent finite
rate of innovation framework, has paved the way to low rate sampling of such
pulses by noticing that only a small number of
parameters per unit time are needed to fully describe these
signals. Unfortunately, for high rates of innovation, existing
sampling schemes are numerically unstable. In this paper we propose a
general sampling approach which leads to stable recovery even in
the presence of many pulses. We begin by deriving a condition on
the sampling kernel which allows perfect reconstruction of
periodic streams from the minimal number of samples. We then design a compactly supported class of filters,
satisfying this condition. The periodic solution is extended to finite and infinite streams,
and is shown to be numerically stable even for a large
number of pulses. High noise robustness is
also demonstrated when the delays are sufficiently
separated. Finally, we process ultrasound imaging data using our techniques, and show that substantial rate reduction with
respect to traditional ultrasound sampling schemes can be achieved.
\end{abstract}


\section{Introduction}
\PARstart{S}{ampling} is the process of representing a continuous-time signal
by discrete-time coefficients, while retaining the important
signal features. The well-known
Shannon-Nyquist theorem states that the minimal sampling rate
required for perfect reconstruction of bandlimited signals is
twice the maximal frequency. This
result has since been generalized to minimal rate sampling schemes for
signals lying in arbitrary subspaces
\cite{BeyondBanlimitedSampling_YoninaTomer2009,TomerYoninaSamplingBookChapter}.

Recently, there has been growing interest in sampling of signals
consisting of a stream of short pulses, where the pulse shape is
known. Such signals have a finite number of degrees of freedom per
unit time, also known as the Finite Rate of Innovation (FRI)
property \cite{Vetterli2002_Basic}. This interest is motivated by
applications such as digital processing of neuronal signals,
bio-imaging, image processing and ultrawideband (UWB)
communications, where such signals are present in abundance. Our
work is motivated by the possible application of this model in
ultrasound imaging, where echoes of the transmit pulse are reflected off scatterers within the tissue, and form a stream of pulses signal at the receiver. The time-delays and amplitudes
of the echoes indicate the position and strength of the various
scatterers, respectively. Therefore, determining these parameters
from low rate samples of the received signal is an important
problem. Reducing the rate allows more efficient processing which
can translate to power and size reduction of the ultrasound
imaging system.

Our goal is to design a minimal rate single-channel sampling and
reconstruction scheme for pulse streams that is stable even in the
presence of many pulses. Since the set of FRI signals does not
form a subspace, classic subspace schemes cannot be directly used
to design low-rate sampling schemes. Mathematically, such FRI signals
conform with a broader model of signals lying in a union of
subspaces \cite{UnionOfSubspaces_LuDo2008,
Yonina2009compressed,RobustRecovery_UnionOfSubspaces_YoninaMishali2009,KfirYonina2009,
mishali2009blind,XamplingPractice}.
Although the minimal sampling rate required for such settings has been
derived, no generic sampling scheme exists for the general
problem. Nonetheless, some special cases have been treated in
previous work, including streams of pulses.

A stream of pulses can be viewed as a parametric signal, uniquely
defined by the time-delays of the pulses and their amplitudes. Efficient sampling of periodic impulse streams, having
$L$ impulses in each period, was proposed in
\cite{Vetterli2002_Basic,VetterliMagazine2008}. The heart of the solution is to obtain a set of Fourier series coefficients, which then converts the
problem of determining the time-delays and amplitudes to that of finding the frequencies and amplitudes of a
sum of sinusoids. The latter is a standard problem in spectral
analysis \cite{Stoica1997} which can be solved using conventional
methods, such as the annihilating filter approach, as long as the
number of samples is at least $2L$. This result is intuitive
since there are $2L$ degrees of freedom in each period: $L$
time-delays and $L$ amplitudes.

Periodic streams of pulses are mathematically convenient to
analyze, however not very practical. In contrast, finite streams
of pulses are prevalent in applications such as ultrasound
imaging. The first treatment of finite Dirac streams appears in
\cite{Vetterli2002_Basic}, in which a Gaussian sampling kernel was
proposed. The time-delays and amplitudes are then estimated from
the Gaussian tails. This method and its improvement
\cite{Vetterli_FRI_InThePresenceOfNoise} are numerically unstable
for high rates of innovation, since they rely on the Gaussian tails which take on small values. The work in \cite{DragottiStrangFix2007} introduced a general family of polynomial and exponential reproducing kernels, which can be used to solve FRI problems. Specifically, B-spline and E-spline sampling kernels which satisfy the reproduction condition are proposed. This method treats streams of Diracs, differentiated
Diracs, and short pulses with compact support. However, the proposed
sampling filters result in poor reconstruction results for large $L$. To the best
of our knowledge, a numerically stable sampling and reconstruction
scheme for high order
problems has not yet been reported.

Infinite streams of pulses arise in applications such as UWB
communications, where the communicated data changes frequently.
Using spline filters \cite{DragottiStrangFix2007},
and under certain limitations on the signal, the infinite stream
can be divided into a sequence of separate finite problems. The
individual finite cases may be treated using methods for the
finite setting, at the expense of above critical sampling rate, and suffer from the same instability issues. In addition, the
constraints that are cast on the signal become more and more
stringent as the number of pulses per unit time grows. In a recent
work \cite{UnserFRI2008} the authors propose a sampling and
reconstruction scheme for $L=1$, however, our interest here is in high values of $L$.

Another related work \cite{KfirYonina2009} proposes a
semi-periodic model, where the pulse time-delays do not change
from period to period, but the amplitudes vary. This is a hybrid
case in which the number of degrees of freedom in the time-delays
is finite, but there is an infinite number of degrees of freedom
in the amplitudes. Therefore, the proposed recovery scheme
generally requires an infinite number of samples. This differs
from the periodic and finite cases we discuss in this paper which
have a finite number of degrees of freedom and, consequently,
require only a finite number of samples.

In this paper we study sampling of signals consisting of a stream
of pulses, covering the three different cases: periodic, finite
and infinite streams of pulses. The criteria we consider for
designing such systems are: a) Minimal sampling rate which allows
perfect reconstruction, b) numerical stability (with sufficiently
separated time delays), and c) minimal restrictions on the number
of pulses per sampling period.

We begin by treating periodic pulse streams. For this setting, we
develop a general sampling scheme for arbitrary pulse shapes which allows to determine the
times and amplitudes of the pulses, from a minimal number of
samples. As
we show, previous work \cite{Vetterli2002_Basic} is a special case
of our extended results. In contrast to the infinite time-support
of the filters in \cite{Vetterli2002_Basic}, we develop a
compactly supported class of filters which satisfy our
mathematical condition. This class of filters consists of a sum of
sinc functions in the frequency domain. We therefore refer to such
functions as \textit{Sum of Sincs} (SoS). To the best of our
knowledge, this is the first class of finite support filters that
solve the periodic case. As we discuss in detail in Section~\ref{SingleCh:sec_related_work}, these filters are related to exponential reproducing kernels, introduced in \cite{DragottiStrangFix2007}.

The compact support of the SoS filters is the key to
extending the periodic solution to the finite stream case.
Generalizing the SoS class, we design a sampling and
reconstruction scheme which perfectly reconstructs a finite stream
of pulses from a minimal number of samples, as long as the pulse
shape has compact support. Our reconstruction is numerically
stable for both small values of $L$ and large number of pulses,
e.g., $L=100$. In contrast, Gaussian sampling filters
\cite{Vetterli2002_Basic,Vetterli_FRI_InThePresenceOfNoise} are
unstable for $L>9$, and we show in simulations that B-splines and E-splines \cite{DragottiStrangFix2007} exhibit large estimation
errors for $L \geq 5$. In addition, we demonstrate substantial
improvement in noise robustness even for low values of $L$. Our
advantage stems from the fact that we propose compactly supported
filters on the one hand, while staying within the regime of
Fourier coefficients reconstruction on the other hand. Extending our results to the infinite setting, we consider an infinite stream consisting of pulse bursts, where
each burst contains a large number of pulses. The stability of our method
allows to reconstruct even a large number
of closely spaced pulses, which cannot be treated using existing
solutions \cite{DragottiStrangFix2007}. In addition, the constraints cast on the
structure of the signal are independent of $L$ (the number of
pulses in each burst), in contrast to previous work, and therefore similar sampling schemes may
be used for different values of $L$. Finally, we show that our
sampling scheme requires lower sampling rate for $L\geq 3$.

As an application, we demonstrate our sampling
scheme on real ultrasound imaging data acquired by GE healthcare's
ultrasound system. We obtain high accuracy estimation while
reducing the number of samples by two orders of magnitude
in comparison with current imaging techniques.

The remainder of the paper is organized as follows. In Section
\ref{SingleCh:sec_periodicFRI} we present the periodic signal model, and
derive a general sampling scheme. The
SoS class is then developed and demonstrated via simulations. The extension to the finite case is presented in Section
\ref{SingleCh:sec_aperiodicFRI}, followed by simulations showing the
advantages of our method in high order problems and noisy
settings. In Section \ref{SingleCh:sec_infiniteFRI}, we treat infinite streams of pulses.
Section~\ref{SingleCh:sec_related_work} explores the relationship of our
work to previous methods. Finally, in
Section~\ref{SingleCh:sec_ultrasound}, we demonstrate our algorithm on real
ultrasound imaging data.

\section{Periodic Stream of Pulses}\label{SingleCh:sec_periodicFRI}
\subsection{Problem Formulation}\label{SingleCh:subsec_periodicFRI_problem_formulation}
Throughout the paper we denote matrices and vectors by bold font, with lowercase letters corresponding to vectors and uppercase letters to
matrices. The $n$th element of a vector $\mathbf{a}$ is written as $\mathbf{a}_{n}$, and $\mathbf{A}_{ij}$ denotes the
$ij$th element of a matrix $\mathbf{A}$. Superscripts $\left(\cdot\right)^{*}$, $\left(\cdot\right)^{T}$ and
$\left(\cdot\right)^{H}$ represent complex conjugation, transposition and conjugate transposition, respectively. The
Moore-Penrose pseudo-inverse of a matrix $\mathbf{A}$ is written as  $\mathbf{A}^{\dagger}$. The continuous-time Fourier transform (CTFT) of a
continuous-time signal $x\left(t\right)\in L_{2}$ is defined by
$X\left(\omega\right)=\int_{-\infty}^{\infty}x\left(t\right)e^{-j\omega t}{\rm d}t$, and
\begin{equation}
\left\langle x\left(t\right),y\left(t\right)\right\rangle =\int_{-\infty}^{\infty}x^*\left(t\right)y\left(t\right){\rm
d}t,
\end{equation}
denotes the inner product between two $L_2$ signals.

Consider a $\tau$-periodic stream of pulses, defined as
\begin{equation}\label{SingleCh:eq_sig_model_periodic}
  x(t) = \sum_{m\in \mathbb{Z}} \sum_{l=1}^L a_l h(t - t_l - m\tau),
\end{equation}
where $h(t)$ is a known pulse shape, $\tau$ is the known period,
and $\{t_l,a_l\}_{l=1}^L,\, t_l \in [0,\tau)$, $a_l \in
\mathbb{C}, \, l=1\ldots L$ are the unknown delays and amplitudes.
Our goal is to sample $x(t)$ and reconstruct it, from a minimal
number of samples. Since the signal has $2L$ degrees of freedom,
we expect the minimal number of samples to be $2L$. We are
primarily interested in pulses which have small time-support.
Direct uniform sampling of $2L$ samples of the signal will result
in many zero samples, since the probability for the sample to hit
a pulse is very low. Therefore, we must construct a more
sophisticated sampling scheme.

Define the periodic continuation of $h(t)$ as $f(t) = \sum_{m \in
\mathbb{Z}} h(t-m\tau)$. Using Poisson's summation formula
\cite{porat1997course}, $f(t)$ may be written as
\begin{equation}\label{SingleCh:eq_f_t}
  f(t) = \frac{1}{\tau} \sum_{k \in \mathbb{Z}} H\left(\frac{2\pi k}{\tau}\right) e^{j2\pi k t/\tau},
\end{equation}
where $H(\omega)$ denotes the CTFT of the pulse $h(t)$.
Substituting \eqref{SingleCh:eq_f_t} into \eqref{SingleCh:eq_sig_model_periodic} we obtain
\begin{eqnarray}
\nonumber
  x(t) &=& \sum_{l=1}^L a_l f(t - t_l) \\
\nonumber
  &=& \sum_{k \in \mathbb{Z}} \left(\frac{1}{\tau}H\left(\frac{2\pi k}{\tau}\right) \sum_{l=1}^L a_l e^{-j2\pi k t_l/\tau}\right)
  e^{j2\pi k t/\tau}  \\
\label{SingleCh:eq_x_t_Fourier_series_rep}
  &=& \sum_{k \in \mathbb{Z}} X[k] e^{j2\pi k t/\tau},
\end{eqnarray}
where we denoted
\begin{equation}\label{SingleCh:eq_X_k}
X[k] = \frac{1}{\tau} H\left(\frac{2\pi k}{\tau}\right) \sum_{l=1}^L a_l e^{-j2\pi k t_l/\tau}.
\end{equation}
The expansion in \eqref{SingleCh:eq_x_t_Fourier_series_rep} is the Fourier series representation of the $\tau$-periodic signal
$x(t)$ with Fourier coefficients given by \eqref{SingleCh:eq_X_k}.

Following \cite{Vetterli2002_Basic}, we now show that once $2L$ or more Fourier coefficients of $x(t)$ are known, we may use conventional tools from spectral
analysis to determine the unknowns $\{t_l,a_l\}_{l=1}^L$. The method by which the Fourier coefficients are
obtained will be presented in subsequent sections.

Define a set $\mathcal{K}$ of $M$ consecutive indices such that $H\left(\frac{2\pi k}{\tau}\right) \neq 0,\, \forall k\in\mathcal{K}$.
We assume such a set exists, which is usually the case for short time-support pulses $h(t)$. Denote by $\mathbf{H}$ the $M
\times M$ diagonal matrix with \textit kth entry $\frac{1}{\tau}H\left(\frac{2\pi k}{\tau}\right)$, and by $\mathbf{V}(\mathbf{t})$
the $M \times L$ matrix with \textit{kl}th element $e^{-j2\pi k t_l/\tau}$, where $\mathbf{t} = \{t_1,\ldots,t_L\}$ is the
vector of the unknown delays. In addition denote by $\mathbf{a}$ the length-$L$ vector whose \textit lth element is $a_l$, and
by $\mathbf{x}$ the length-$M$ vector whose \textit kth element is $X[k]$. We may then write \eqref{SingleCh:eq_X_k} in matrix form as
\begin{equation}\label{SingleCh:eq_X_k_matrix}
\mathbf{x} = \mathbf{H} \mathbf{V}(\mathbf{t}) \mathbf{a}.
\end{equation}
Since $\mathbf{H}$ is invertible by construction we define $\mathbf{y} = \mathbf{H}^{-1} \mathbf{x}$, which satisfies
\begin{equation}\label{SingleCh:eq_y_equal_x_deconvolved}
  \mathbf{y} = \mathbf{V}(\mathbf{t}) \mathbf{a}.
\end{equation}
The matrix $\mathbf{V}$ is a Vandermonde matrix and therefore has full column rank \cite{Stoica1997,hoffman1971linear} as long
as $M \geq L$ and the time-delays are distinct, i.e., $t_i \neq t_j$ for all $i \neq j$.

Writing the expression for the $k$th element of the vector $\mathbf{y}$ in \eqref{SingleCh:eq_y_equal_x_deconvolved} explicitly:
\begin{equation}
  \mathbf{y}_k = \sum_{l=1}^L a_l e^{-j2\pi k t_l/\tau}.
\end{equation}
Evidently, given the vector $\mathbf{x}$, \eqref{SingleCh:eq_y_equal_x_deconvolved} is a standard problem of finding the frequencies and
amplitudes of a sum of $L$ complex exponentials (see \cite{Stoica1997} for a review of this topic). This problem may be
solved as long as $|\mathcal{K}|=M \geq 2L$.

The annihilating filter approach used extensively by Vetterli et
al. \cite{Vetterli2002_Basic},\cite{VetterliMagazine2008} is one
way of recovering the frequencies, and is thoroughly described in the literature \cite{Stoica1997,Vetterli2002_Basic,VetterliMagazine2008}. This method can solve the
problem using the critical number of samples $M=2L$, as opposed to
other techniques such as MUSIC
\cite{MUSIC_Schmidt},\cite{MUSIC_Bienvenu} and ESPRIT
\cite{ESPRIT_Kailath} which require oversampling.
Since we are interested in minimal-rate sampling, we use the
annihilating filter throughout the paper.

\subsection{Obtaining The Fourier Series Coefficients}\label{SingleCh:subsec_general_single_channel_sampling}
As we have seen, given the vector of $M\geq2L$ Fourier series coefficients $\mathbf{x}$, we may use standard tools from
spectral analysis to determine the set $\{t_l,a_l\}_{l=1}^L$. In practice, however, the signal is sampled in the time domain, and therefore we do not have direct access to samples of $\mathbf{x}$. Our goal is to design a
single-channel sampling scheme which allows to determine $\mathbf{x}$ from time-domain samples.
In contrast to previous work \cite{Vetterli2002_Basic,VetterliMagazine2008} which focused on a low-pass sampling filter
, in this section we
derive a general condition on the sampling kernel allowing to obtain the vector $\mathbf{x}$. For the sake of clarity we confine ourselves to uniform sampling, the results extend in a straightforward manner to nonuniform sampling as well.

\RonenComment{
\begin{figure}[h]
\centering
\includegraphics[scale=0.8]{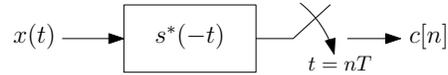}
\caption{Single channel sampling scheme.} \label{SingleCh:fig_basic_sampling}
\end{figure}}
Consider sampling the signal $x(t)$ uniformly with sampling kernel $s^*(-t)$ and sampling period $T$, as depicted in
Fig.\RonenComment{~\ref{SingleCh:fig_basic_sampling}}. The samples are given by
\begin{equation}\label{SingleCh:eq_c_n_basic_samples}
  c[n] = \int_{-\infty}^{\infty} x(t) s^*(t-nT) {\rm d}t = \langle s(t-nT),x(t) \rangle.
\end{equation}
Substituting \eqref{SingleCh:eq_x_t_Fourier_series_rep} into \eqref{SingleCh:eq_c_n_basic_samples} we have
\begin{eqnarray}
\nonumber c[n] &=& \sum_{k\in\mathbb{Z}}X[k] \int_{-\infty}^{\infty} e^{j2\pi k
t/\tau}s^*(t-nT)dt \\
\nonumber
&=& \sum_{k\in\mathbb{Z}} X[k] e^{j2\pi k n T/\tau} \int_{-\infty}^{\infty} e^{j2\pi kt/\tau}s^*(t)dt \\
&=& \sum_{k\in\mathbb{Z}} X[k] e^{j2\pi k n T/\tau} S^*(2\pi k/\tau), \label{SingleCh:eq_c_n}
\end{eqnarray}
where $S(\omega)$ is the CTFT of $s(t)$. Choosing any filter $s(t)$ which satisfies
\begin{equation}\label{SingleCh:eq_S_omega}
S(\omega) = \left\{
\begin{array}{l l}
  0 & \omega = 2\pi k/\tau, k \notin \mathcal{K} \\
  \mbox{nonzero} & \omega = 2\pi k/\tau, k \in \mathcal{K} \\
  \textrm{arbitrary} & \quad \mbox{otherwise},\\
\end{array} \right.
\end{equation}
we can rewrite \eqref{SingleCh:eq_c_n} as
\begin{equation}\label{SingleCh:eq_c_n_good_s_t}
c[n] = \sum_{k\in\mathcal{K}} X[k] e^{j2\pi k n T/\tau} S^*(2\pi k/\tau).
\end{equation}
In contrast to \eqref{SingleCh:eq_c_n}, the sum in \eqref{SingleCh:eq_c_n_good_s_t} is finite. Note that \eqref{SingleCh:eq_S_omega} implies that any
real filter meeting this condition will satisfy $k \in \Kset \Rightarrow -k \in \Kset$, and in addition $S(2\pi k/\tau) =
S^*(-2\pi k/\tau)$, due to the conjugate symmetry of real filters.

Defining the $M \times M$ diagonal matrix $\mathbf{S}$ whose \textit kth entry is $S^*(2\pi k/\tau)$ for all $k \in
\mathcal{K}$, and the length-$N$ vector $\mathbf{c}$ whose \textit nth element is $c[n]$, we may write \eqref{SingleCh:eq_c_n_good_s_t}
as
\begin{equation}\label{SingleCh:eq_c_n_matrix}
\mathbf{c} = \mathbf{V}(-\mathbf{t}_s) \mathbf{S} \mathbf{x}
\end{equation}
where $\mathbf{t}_s = \{nT: n=0 \ldots N-1\}$, and $\mathbf{V}$ is defined as in \eqref{SingleCh:eq_X_k_matrix} with a different parameter
$-\mathbf{t}_s$ and dimensions $N \times M$. The matrix $\mathbf{S}$ is invertible by construction. Since $\mathbf{V}$ is Vandermonde, it is left invertible as long as $N \geq M$. Therefore,
\begin{equation}\label{SingleCh:eq_obtaining_x_vector_from_samples}
  \mathbf{x} = \mathbf{S}^{-1} \mathbf{V}^{\dag}(-\mathbf{t}_s) \mathbf{c}.
\end{equation}
In the special case where $N=M$ and $T = \tau/N$, the recovery in \eqref{SingleCh:eq_obtaining_x_vector_from_samples} becomes:
\begin{equation}\label{SingleCh:eq_X_k_matrix_DFT}
  \mathbf{x} = \mathbf{S}^{-1} \textrm{DFT}\{\mathbf{c}\},
\end{equation}
i.e., the vector $\mathbf{x}$ is obtained by applying the Discrete Fourier Transform (DFT) on the sample vector, followed by a
correction matrix related to the sampling filter.

The idea behind this sampling scheme is that each sample is actually a linear combination of the elements of $\mathbf{x}$. The
sampling kernel $s(t)$ is designed to pass the coefficients $X[k],\, k \in \mathcal{K}$ while suppressing all other
coefficients $X[k],\, k \notin \mathcal{K}$. This is exactly what the condition in \eqref{SingleCh:eq_S_omega} means. This sampling
scheme guarantees that each sample combination is linearly independent of the others. Therefore, the linear system of
equations in \eqref{SingleCh:eq_c_n_matrix} has full column rank which allows to solve for the vector $\mathbf{x}$.

We summarize this result in the following theorem.
\begin{theorem}\label{SingleCh:prop_periodic_general_codition_filter}
  Consider the $\tau$-periodic stream of pulses of order $L$:
  \begin{equation*}
    x(t) = \sum_{m\in \mathbb{Z}} \sum_{l=1}^L a_l h(t - t_l - m\tau).
  \end{equation*}
  Choose a set
  $\mathcal{K}$ of consecutive indices for which $H(2\pi k/\tau) \neq 0,\, \forall k\in\mathcal{K}$. Then the samples
  \begin{equation*}
    c[n] = \langle s(t-nT),x(t) \rangle , \quad n=0\ldots N-1,
  \end{equation*}
  uniquely determine the signal $x(t)$ for any $s(t)$ satisfying condition \eqref{SingleCh:eq_S_omega}, as long as
  $N \geq |\mathcal{K}| \geq 2L$.
\end{theorem}
In order to extend Theorem~\ref{SingleCh:prop_periodic_general_codition_filter} to nonuniform sampling, we only need to
substitute the nonuniform sampling times in the vector $\mathbf{t}_s$ in \eqref{SingleCh:eq_obtaining_x_vector_from_samples}.

Theorem~\ref{SingleCh:prop_periodic_general_codition_filter} presents a general single channel sampling scheme. One special
case of this framework is the one proposed by Vetterli et al. in~\cite{Vetterli2002_Basic} in which $s^*(-t) =
B\sinc(-Bt)$, where $B = M/\tau$ and $N \geq M \geq 2L$. In this case $s(t)$ is an ideal
low-pass filter of bandwidth $B$ with
\begin{equation}\label{SingleCh:eq_Vetterli_lowpass}
  S(\omega) = \frac{1}{\sqrt{2\pi}}\rect\left(\frac{\omega}{2\pi B}\right).
\end{equation}
Clearly, \eqref{SingleCh:eq_Vetterli_lowpass} satisfies the general condition in \eqref{SingleCh:eq_S_omega} with $\mathcal{K}=\{-\lfloor
M/2 \rfloor,\ldots,\lfloor M/2 \rfloor\}$ and $S\left(\frac{2\pi k}{\tau}\right) = \frac{1}{\sqrt{2\pi}},\, \forall k \in \mathcal{K}$. Note
that since this filter is real valued it must satisfy $k\in \Kset \Rightarrow -k\in\Kset$, i.e., the indices come in pairs
except for $k=0$. Since $k=0$ is part of the set $\Kset$, in this case the cardinality $M=|\Kset|$ must be odd valued so that $N \geq M \geq 2L+1$ samples, rather than the minimal rate $N \geq 2L$.

The ideal low-pass filter is bandlimited, and therefore has infinite time-support, so that it cannot be extended to finite and infinite streams of pulses. In the next section we propose a class of
non-bandlimited sampling kernels, which exploit the additional degrees of freedom in condition \eqref{SingleCh:eq_S_omega}, and
have compact support in the time domain. The compact support allows to extend this class to finite and infinite
streams, as we show in Sections \ref{SingleCh:sec_aperiodicFRI} and \ref{SingleCh:sec_infiniteFRI}, respectively.

\subsection{Compactly Supported Sampling Kernels}
\RonenComment{
\begin{figure*}
\centering \mbox { \subfigure{\includegraphics[scale=0.7]{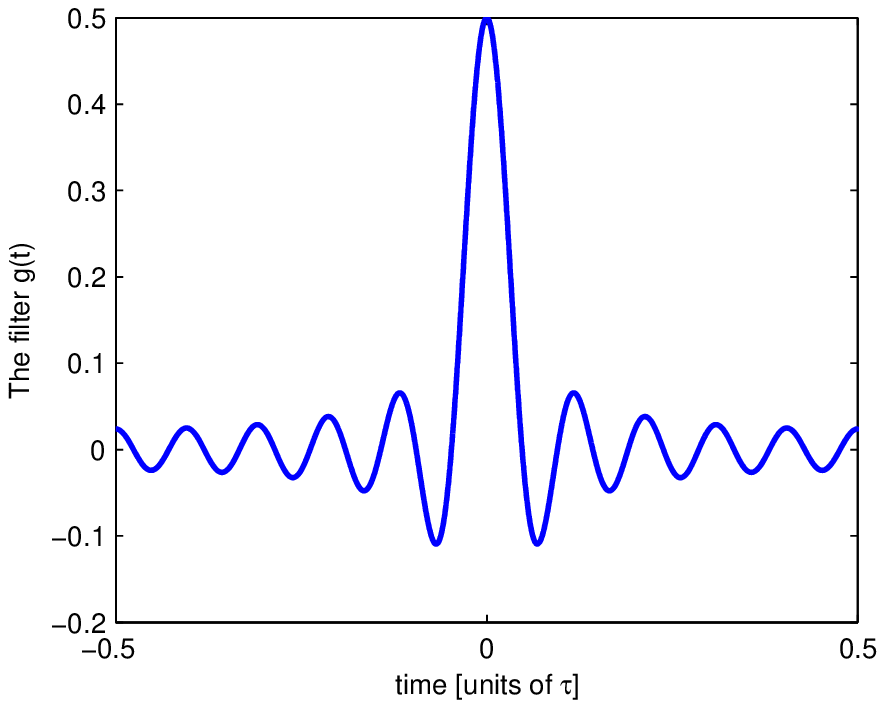}}
\subfigure{\includegraphics[scale=0.7]{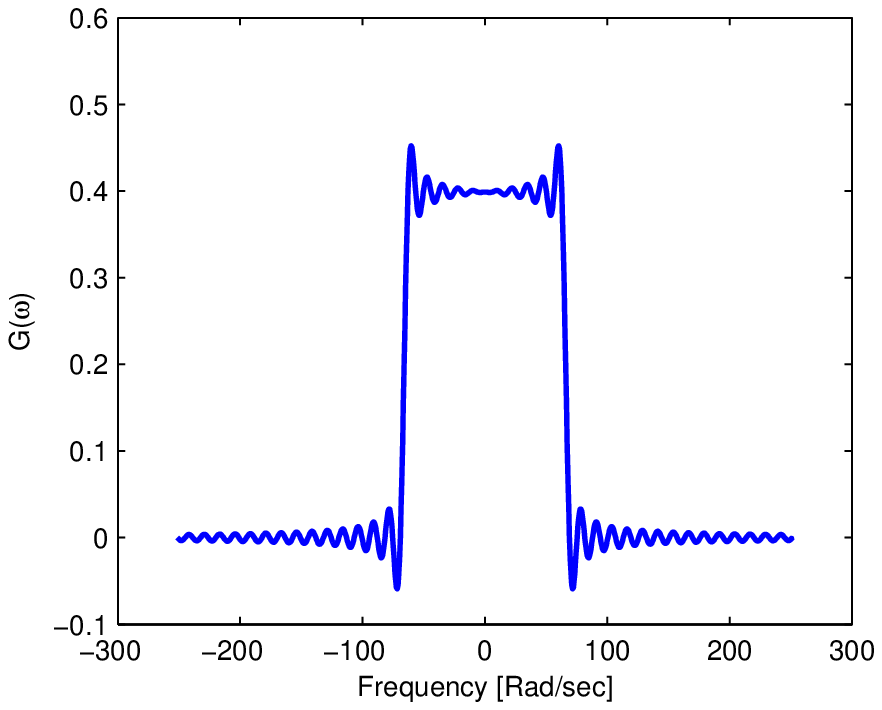}} }
\caption{The filter $g(t)$ with all coefficients $b_k=1$.} \label{SingleCh:fig_all_ones_filter}
\end{figure*}}
\RonenComment{
\begin{figure*}
\centering \mbox { \subfigure{\includegraphics[scale=0.7]{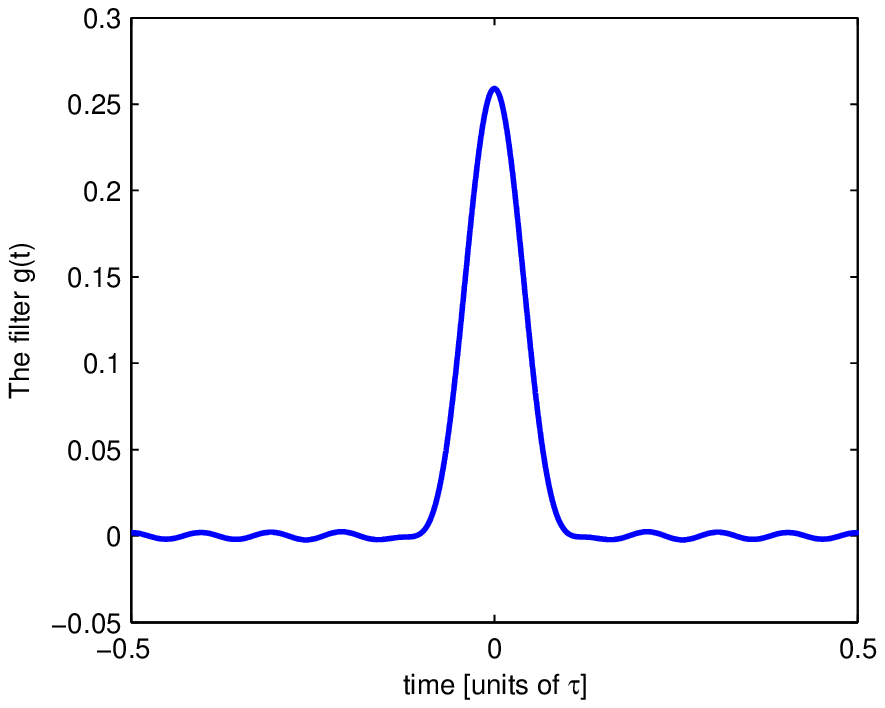}} \subfigure{\includegraphics[scale=0.7]{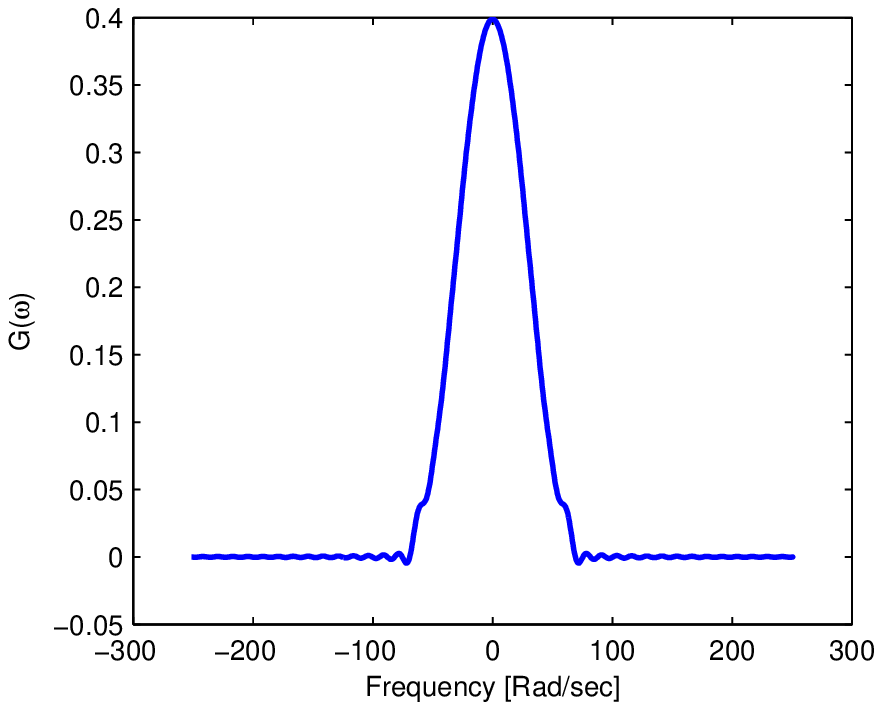}} } \caption{The filter
$g(t)$ with Hamming window coefficients.} \label{SingleCh:fig_hamming_filter}
\end{figure*}}
Consider the following SoS class which consists of a sum of sincs in the frequency domain:
\begin{equation}\label{SingleCh:eq_G_omega}
  G(\omega) = \frac{\tau}{\sqrt{2\pi}} \sum_{k \in \mathcal{K}} b_k \sinc\left(\frac{\omega}{2\pi/\tau} - k\right)
\end{equation}
where $b_k \neq 0,\, k \in \mathcal{K}$. The filter in \eqref{SingleCh:eq_G_omega} is real valued if and only if $k\in\Kset
\Rightarrow -k \in\Kset$ and $b_k = b^*_{-k}$ for all $k\in\Kset$. Since for each sinc in the sum
\begin{equation}
  \sinc\left(\frac{\omega}{2\pi/\tau} - k\right) = \left\{
\begin{array}{l l}
  1 & \omega = 2\pi k'/\tau,\, k'=k \\
  0 & \omega = 2\pi k'/\tau,\, k' \neq k, \\
\end{array} \right.
\end{equation}
the filter $G(\omega)$ satisfies \eqref{SingleCh:eq_S_omega} by construction. Switching to the time domain
\begin{equation}\label{SingleCh:eq_g_t}
  g(t) = \rect\left(\frac{t}{\tau}\right) \sum_{k \in \mathcal{K}} b_k e^{j2\pi kt/\tau},
\end{equation}
which is clearly a time compact filter with support $\tau$.

The SoS class in \eqref{SingleCh:eq_g_t} may be extended to
\begin{equation}\label{SingleCh:eq_G_omega_extended}
  G(\omega) = \frac{\tau}{\sqrt{2\pi}} \sum_{k \in \mathcal{K}} b_k \phi\left(\frac{\omega}{2\pi/\tau} - k\right)
\end{equation}
where $b_k \neq 0,\, k \in \mathcal{K}$, and $\phi(\omega)$ is any function satisfying:
\begin{equation}
  \phi\left(\omega \right) = \left\{
\begin{array}{l l}
  1 & \omega = 0 \\
  0 & |\omega| \in \mathbb{N} \\
  \textrm{arbitrary} & \mbox{otherwise}. \\
\end{array} \right.
\end{equation}
This more general structure allows for smooth versions of the rect function, which is important when practically implementing analog filters.

\RonenComment{
\begin{figure*}
\centering \mbox { \subfigure[Periodic signal]{\includegraphics[scale=0.7]{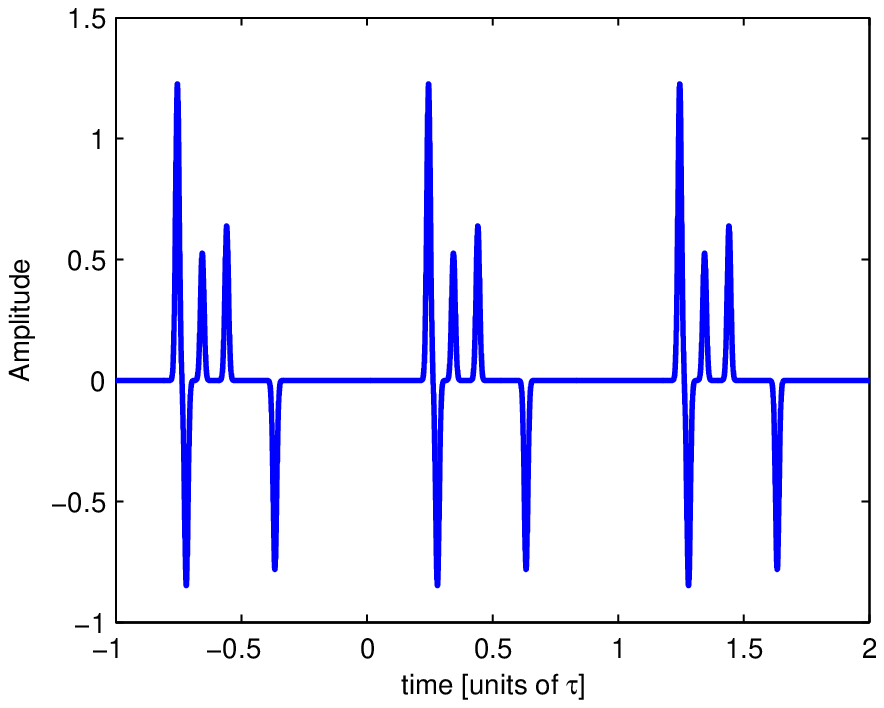}}
} \\
\centering \mbox { \subfigure[Sampling
filter]{\includegraphics[scale=0.7]{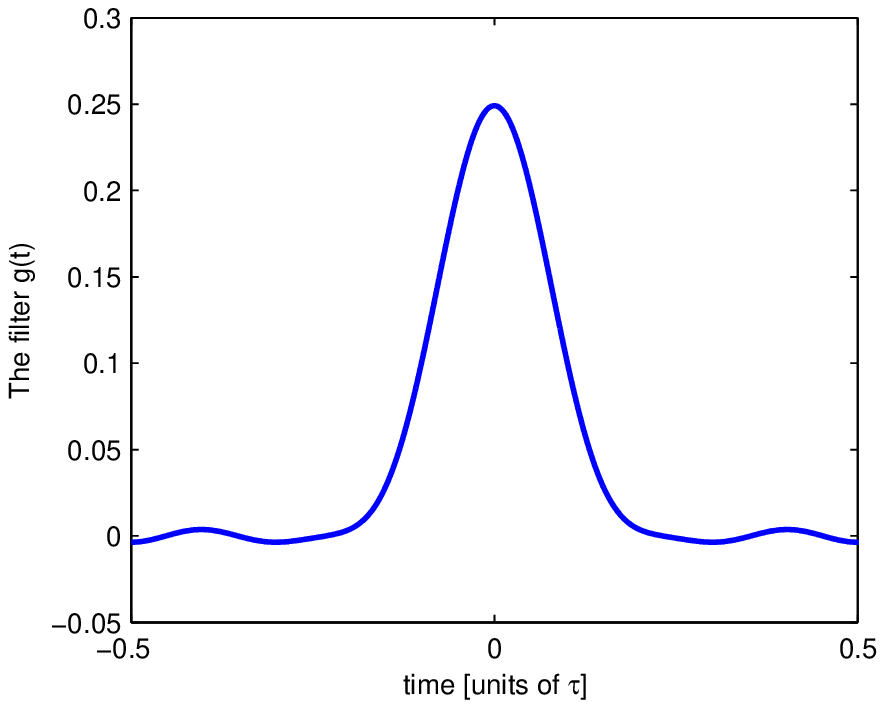}}
\subfigure[Low rate
samples]{\includegraphics[scale=0.7]{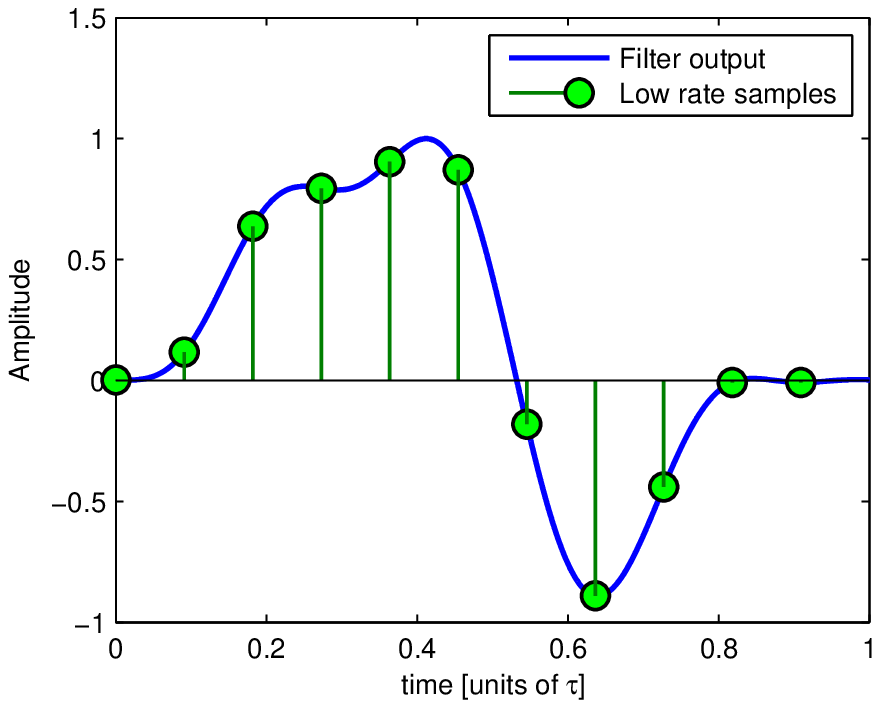}}}
\caption{Compressed samples of pulse streams (a) Original
periodic signal consisting of 5 Gaussians (3 periods
are shown). (b) Sampling filter. (c) Low rate samples depicted
over the filtered signal.}
\label{SingleCh:fig_periodic_demonstration_sampling}
\end{figure*}}
The function $g(t)$ represents a class of filters determined by the parameters $\{b_k\}_{k\in\mathcal K}$. These degrees
of freedom offer a filter design tool where the free parameters $\{b_k\}_{k\in\mathcal K}$ may be optimized
for different goals, e.g., parameters which will result in a feasible analog filter. In Theorem~\ref{SingleCh:theorem_optimal_b_k} below, we show how to choose $\{b_k\}$ to minimize the mean-squared error (MSE) in the presence of noise.

Determining the parameters $\{b_k\}_{k\in\mathcal K}$ may be viewed from a more empirical point of view. The impulse response of any analog filter having support $\tau$ may be
written in terms of a windowed Fourier series as
\begin{equation}
  \Phi(t) = \rect\left(\frac{t}{\tau}\right) \sum_{k \in \mathbb{Z}} \beta_k e^{j2\pi kt/\tau}.
\end{equation}
\RonenComment{
\begin{figure*}
\centering \mbox {
\subfigure[]{\includegraphics[scale=0.8]{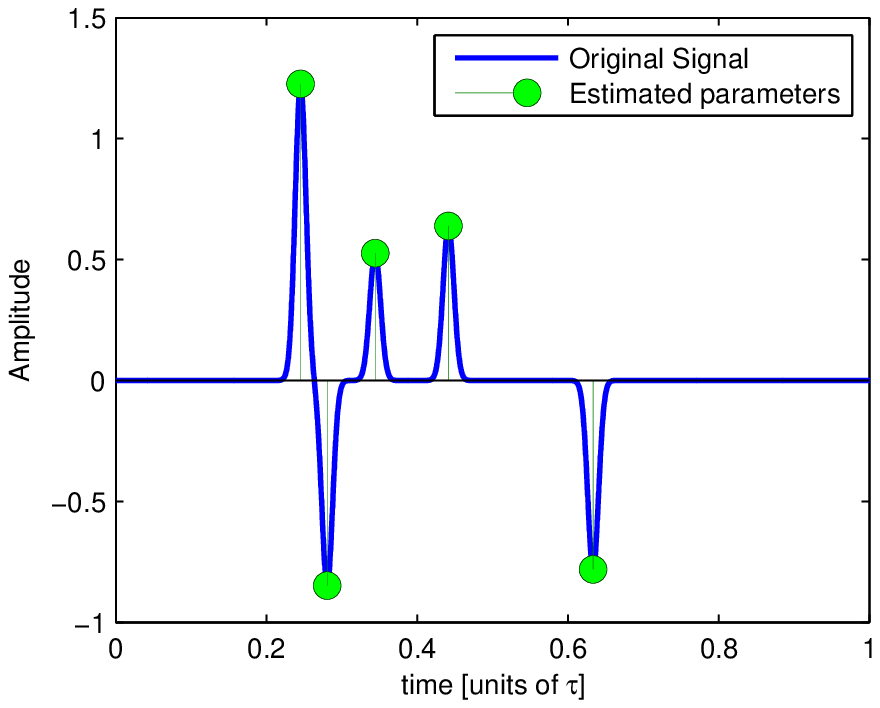}}
\subfigure[]{\includegraphics[scale=0.8]{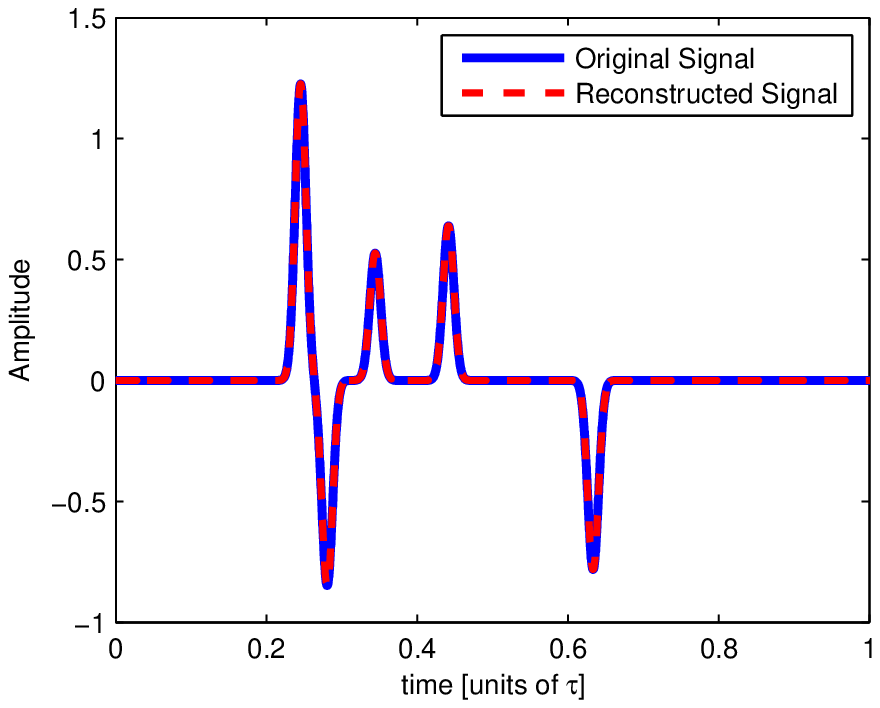}}}
\caption{(a) Estimated time-delays and amplitudes depicted over the original signal. (b) Reconstructed signal vs. original
one. The reconstruction is exact to numerical precision.} \label{SingleCh:fig_periodic_demonstration_reconstruction}
\end{figure*}}Confining ourselves to filters which satisfy $\beta_k \neq 0,\, k\in\mathcal{K}$, we may truncate the series and choose:
\begin{equation}
  b_k = \left\{
\begin{array}{l l}
  \beta_k & k \in \mathcal{K} \\
  0 & k \notin \mathcal{K} \\
\end{array} \right.
\end{equation}
as the parameters of $g(t)$ in \eqref{SingleCh:eq_g_t}. With this choice, $g(t)$ can be viewed as an approximation to $\Phi(t)$. Notice that there is an inherent tradeoff here: using more coefficients will result in a better approximation
of the analog filter, but in turn will require more samples, since the number of samples $N$ must be greater than the cardinality of the set $\mathcal{K}$.

To demonstrate the filter $g(t)$ we first choose $\mathcal{K} =
\{-p,\ldots,p\}$ and set all coefficients $\{b_k\}$ to one,
resulting in
\begin{equation}\label{SingleCh:eq_Dirichlet_filter}
  g(t) = \rect\left(\frac{t}{\tau}\right) \sum_{k=-p}^p e^{j2\pi kt/\tau} = \rect\left(\frac{t}{\tau}\right) D_p(2\pi
  t/\tau),
\end{equation}
where the Dirichlet kernel $D_p(t)$ is defined by
\begin{equation}
  D_p(t) = \sum_{k=-p}^p e^{jkt} = \frac{\sin\left(\left(p +\frac{1}{2}\right)t\right)}{\sin(t/2)}.
\end{equation}
The resulting filter for $p=10$ and $\tau = 1 \textrm{ sec}$, is depicted in
Fig.\RonenComment{~\ref{SingleCh:fig_all_ones_filter}}. This filter is also optimal in an MSE sense for the case $h(t)=\delta(t)$, as we show in Theorem~\ref{SingleCh:theorem_optimal_b_k}. In Fig.\RonenComment{~\ref{SingleCh:fig_hamming_filter}} we plot $g(t)$ for the case in which the $b_k$'s are chosen as a
length-$M$ symmetric Hamming window:
\begin{equation}\label{SingleCh:eq_hamming_window}
  b_k = 0.54 - 0.46\cos\left(2\pi\frac{k+\lfloor M/2 \rfloor}{M}\right),\quad k\in\mathcal{K}.
\end{equation}
Notice that in both cases the coefficients satisfy $b_k = b^*_{-k}$, and therefore, the resulting filters
are real valued.

In the presence of noise, the choice of $\{b_k\}_{k \in \mathcal{K}}$ will effect the performance. Consider the case in which digital noise is added to the samples $\mathbf{c}$, so that $\mathbf{y}=\mathbf{c}+\mathbf{w}$, with $\mathbf{w}$ denoting a white Gaussian noise vector. Using \eqref{SingleCh:eq_c_n_matrix},
\begin{equation}\label{SingleCh:eq_c_n_matrix_noisy}
\mathbf{y} = \mathbf{V}(-\mathbf{t}_s) \mathbf{B} \mathbf{x} + \mathbf{w}
\end{equation}
where $\mathbf{B}$ is a diagonal matrix, having $\{b_k\}$ on its diagonal. To choose the optimal $\mathbf{B}$ we assume that the $\{a_l\}$ are uncorrelated with variance $\sigma_a^2$, independent of $\{t_l\}$, and that $\{t_l\}$ are uniformly distributed in $[0,\tau)$. Since the noise is added to the samples after filtering, increasing the filter's amplification will always reduce the MSE. Therefore, the filter's energy must be normalized, and we do so by adding the constraint $\textrm{Tr}(\mathbf{B}^*\mathbf{B}) = 1$. Under these assumptions, we have the following theorem:
\begin{theorem}\label{SingleCh:theorem_optimal_b_k}
  The minimal MSE of a linear estimator of $\mathbf{x}$ from the noisy samples $\mathbf{y}$ in \eqref{SingleCh:eq_c_n_matrix_noisy} is achieved by choosing the coefficients
  \begin{equation}\label{SingleCh:eq_theorem_optimal_b_k}
      |b_i|^2 = \left\{
      \begin{array}{l l}
        \frac{\sigma^2}{N}\left(\sqrt{\frac{N}{\lambda \sigma^2}}-\frac{1}{|\tilde h_i|^2}\right) & \lambda \leq |\tilde h_i|^4 N/\sigma^2 \\
        0 & \lambda > |\tilde h_i|^4 N/\sigma^2\\
      \end{array} \right.
  \end{equation}
  where $\tilde h_k = H(2\pi k/\tau) \sigma_a \sqrt{L} / \tau$ and are arranged in an increasing order of $|\tilde h_k|$,
  \begin{equation}
    \sqrt{\lambda} = \frac{(|\mathcal{K}|-m)\sqrt{N/\sigma^2}}{N/\sigma^2 + \displaystyle\sum_{i=m+1}^{|\mathcal{K}|}1/|\tilde h_i|^2},
  \end{equation}
  and $m$ is the smallest index for which $\lambda \leq |\tilde h_{m+1}|^4 N/\sigma^2$.
\end{theorem}
\begin{proof}
  See the Appendix.
\end{proof}
An important consequence of Theorem 2 is the following corollary.
\begin{corollary}\label{SingleCh:theorem_optimal_b_k_delta}
  If $|\tilde h_k|^2=|\tilde h_\ell|^2,\,\forall k,\ell\in\mathcal{K}$ then the optimal coefficients are $|b_i|^2 = 1/|\mathcal{K}|,\,\forall k\in\mathcal{K}$.
\end{corollary}
\begin{proof}
  It is evident from \eqref{SingleCh:eq_theorem_optimal_b_k} that if $|\tilde h_k|=|\tilde h_\ell|$ then $|b_k| = |b_{\ell}|$. To satisfy the trace constraint $\textrm{Tr}(\mathbf{B}^*\mathbf{B}) = 1$, $\lambda$ cannot be chosen such that all $b_i = 0$. Therefore, $|b_i|^2 = 1/|\mathcal{K}|$ for all $i\in\mathcal{K}$.
\end{proof}
From Corollary~\ref{SingleCh:theorem_optimal_b_k_delta} it follows that when
$h(t) = \delta(t)$, the optimal choice of coefficients is $b_k = b_j$ for all $k$ and $j$. We therefore use this choice when simulating noisy settings in the next section.

Our sampling scheme for the periodic case consists of sampling
kernels having compact support in the time domain. In the next section we exploit the compact support of our filter, and
extend the results to the finite stream case. We will show that our sampling and reconstruction scheme offers a
numerically stable solution, with high noise robustness.

\subsection{Simulations}
\subsubsection{Demonstration of Our Sampling Scheme}
To demonstrate our results, we consider an input $x(t)$ consisting of $L=5$ delayed and weighted
versions of a Gaussian pulse
\begin{equation}\label{SingleCh:eq_gaussian_pulse}
  h(t) = \frac{1}{\sqrt{2\pi\sigma^2}} \exp(-t^2/2\sigma^2),
\end{equation}
with parameter $\sigma = 7\cdot 10^{-3}$, and period $\tau=1$. The time-delays and amplitudes were chosen randomly.
In order to demonstrate near-critical sampling we choose the set of indices $\mathcal{K}=\{-L,\ldots,L\}$ with cardinality
$M=|\mathcal{K}|=11$. We filter $x(t)$ with $g(t)$ of \eqref{SingleCh:eq_hamming_window}. The filter output is sampled uniformly $N$ times, with sampling period $T = \tau/N$, where $N=M=11$. The sampling process is
depicted in Fig.\RonenComment{~\ref{SingleCh:fig_periodic_demonstration_sampling}}. The vector $\mathbf{x}$ is obtained using \eqref{SingleCh:eq_obtaining_x_vector_from_samples}, and the delays and amplitudes are determined by the annihilating filter method. Reconstruction results are depicted in
Fig.\RonenComment{~\ref{SingleCh:fig_periodic_demonstration_reconstruction}}. The estimation and reconstruction are both exact to numerical precision.

Analog filtering operations are carried out by discrete approximations over a fine grid. The analog signal and filters are mimicked by high rate digital signals. Since the sampling rate which constructs the fine grid is between 2-3 orders of magnitude higher than the final sampling rate $T$, the simulations reflect very well the analog results.
\subsubsection{Noisy Case}\label{SingleCh:subsubsec_periodic_noisy}
We now consider the case in which the samples are corrupted by noise. Our signal consists of $L=2$ pulses $h(t) = \delta(t)$. The period was set to $\tau
= 1$, $\mathcal{K} = \{-2,\ldots,2\}$, and $N=M=5$ samples were taken, sampled uniformly with sampling
period $T=\tau/N$. We choose $g(t)$ given by \eqref{SingleCh:eq_Dirichlet_filter}. As explained earlier, only the values of the filter at points $2\pi k/\tau,\,k\in\mathcal{K}$ affect the samples (see \eqref{SingleCh:eq_S_omega}).
Since the values of the filter at the relevant points coincide and are equal to one for the low-pass filter \cite{Vetterli2002_Basic} and $g^*(-t)$, the resulting samples for both settings are identical. Therefore, we present results for our method only, and state that the exact same results are obtained using the approach of \cite{Vetterli2002_Basic}.

In our setup white Gaussian noise (AWGN) with variance $\sigma_n^2$ is added to the samples, where we define the SNR as:
\begin{equation}\label{SingleCh:eq_SNR_definition}
  \mathrm{SNR} = \frac{\frac{1}{N}\|\mathbf{c}\|_2^2}{\sigma_n^2},
\end{equation}
with $\mathbf{c}$ denoting the clean samples. In our experiments the noise variance is set to give the desired SNR.

The simulation consists of $1000$ experiments for each SNR, where in each experiment a new noise vector is created. We choose $\mathbf{t} = \tau\cdot(1/3\,\,2/3)^T$ and $\mathbf{a} = \tau\cdot(1\,\,1)^T$, where these vectors remain constant throughout the experiments. We define the error in time-delay estimation as the average of $\| \mathbf{t} - \hat{\mathbf{t}} \|_2^2$, where $\mathbf{t}$ and $\hat{\mathbf{t}}$ denote the true and estimated time-delays, respectively, sorted in increasing order. The error in amplitudes is similarly defined by $\| \mathbf{a} - \hat{\mathbf{a}} \|_2^2$. In
Fig.\RonenComment{~\ref{SingleCh:fig_periodic_noisy_case}} we show the error as a function of SNR for both delay and amplitude estimation. Estimation of the time-delays is the main interest in FRI literature, due to special nonlinear methods required for delay recovery. Once the delays are known, the standard least-squares method is typically used to recover the amplitudes, therefore, we focus on delay estimation in the sequel.
\begin{figure*}
\centering \mbox { \subfigure[]{\includegraphics[scale=0.8]{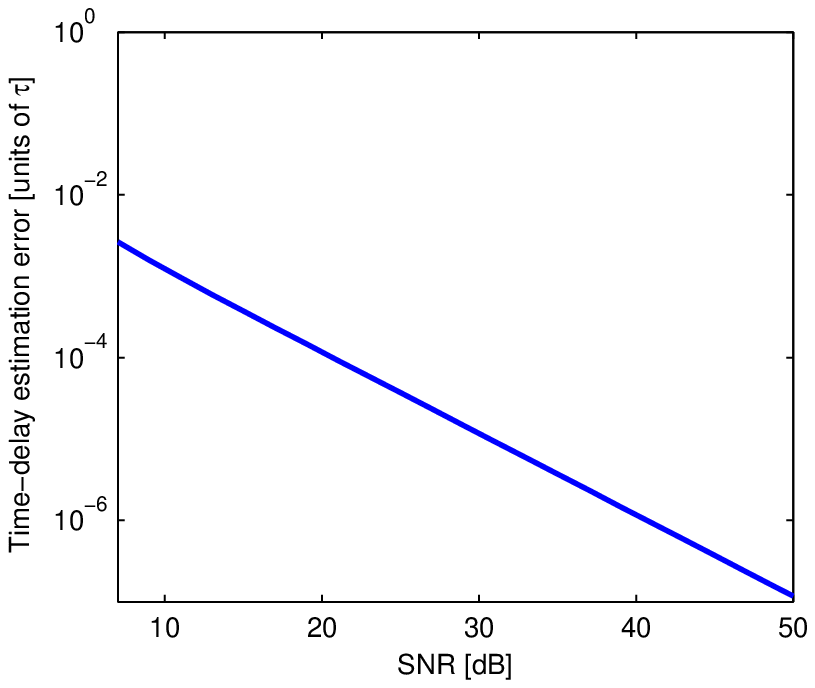}}
\subfigure[]{\includegraphics[scale=0.8]{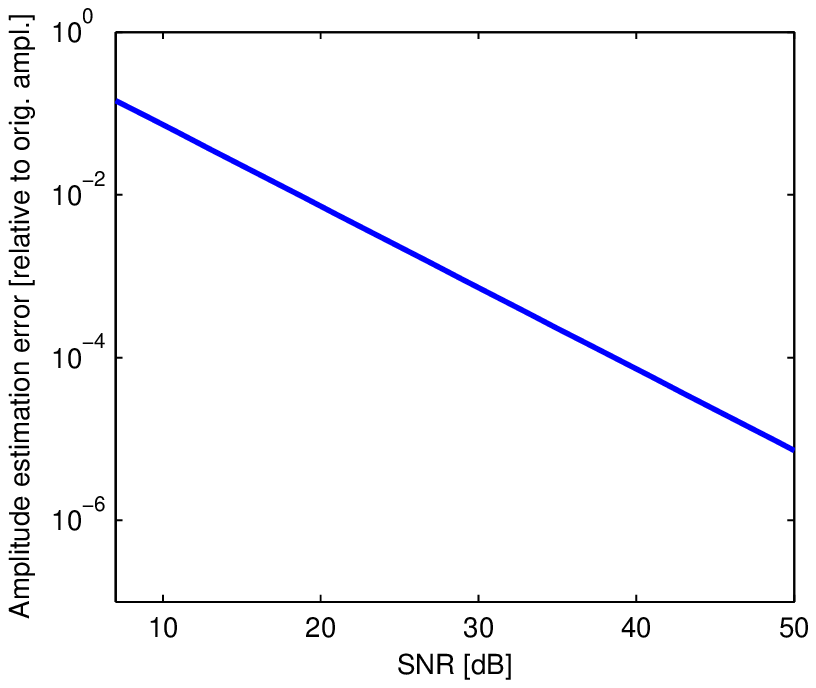}} }
\caption{Performance as a function of SNR, using our periodic approach. Estimation error in (a) delays, and (b) amplitudes.} \label{SingleCh:fig_periodic_noisy_case}
\end{figure*}

Finally, for the same setting we can improve reconstruction accuracy at the expense of
oversampling, as illustrated in
Fig.~\ref{SingleCh:fig_oversampling_comparison}. Here we show recovery
performance for oversampling factors of 1, 2, 4 and 8. The oversampling was exploited using the total least-squares method, followed by Cadzow's iterative denoising (both described in detail in \cite{VetterliMagazine2008}). \RonenComment{
\begin{figure}[h]
\centering
\includegraphics[scale=0.9]{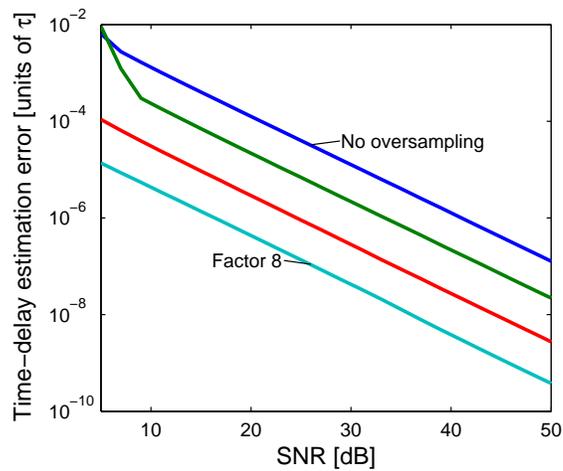}
\caption{The effect of oversampling on estimation error.
Oversampling by a factor of 1, 2, 4 and 8.}
\label{SingleCh:fig_oversampling_comparison}
\end{figure}}

\section{Finite Stream of Pulses}\label{SingleCh:sec_aperiodicFRI}
\subsection{Extension of SoS Class}
Consider now a finite stream of pulses, defined as
\begin{equation}\label{SingleCh:eq_sig_model_aperiodic}
  \tilde x(t) = \sum_{l=1}^L a_l h(t - t_l),\quad t_l \in [0,\tau),\, a_l \in \mathbb{R}, \, l=1\ldots L,
\end{equation}
where, as in Section \ref{SingleCh:sec_periodicFRI}, $h(t)$ is a known pulse shape, and $\{t_l,a_l\}_{l=1}^L$ are the unknown
delays and amplitudes. The time-delays $\{t_l\}_{l=1}^L$ are restricted to lie in a finite time interval $[0,\tau)$.
Since there are only $2L$ degrees of freedom, we wish to design a
sampling and reconstruction method which perfectly reconstructs $\tilde x(t)$ from $2L$ samples. In this section we assume that the pulse $h(t)$ has finite support $R$, i.e.,
\begin{equation}\label{SingleCh:eq_h_t_support_tau_div2}
  h(t) = 0,\, \forall |t| \geq R/2.
\end{equation}
This is a rather weak condition, since our primary interest is in very short pulses which have wide, or even infinite,
frequency support, and therefore cannot be sampled efficiently using classical sampling results for bandlimited signals. We now investigate the structure of the samples taken in the periodic case, and design a sampling kernel for the finite setting which obtains precisely the same samples $c[n]$, as in the periodic case.

In the periodic setting, the resulting samples are given by \eqref{SingleCh:eq_c_n}. Using $g(t)$ of \eqref{SingleCh:eq_g_t} as the sampling kernel we have
\begin{eqnarray}\label{SingleCh:eq_periodic_case_samples_phi_introduced}
\nonumber
c[n] &=& \langle g(t-nT),x(t) \rangle \\
\nonumber
&=& \sum_{m\in\mathbb{Z}} \sum_{l=1}^L a_l \int_{-\infty}^{\infty} h(t-t_l-m\tau)g^*(t-nT)dt \\
\nonumber
&=& \sum_{m\in\mathbb{Z}} \sum_{l=1}^L a_l \int_{-\infty}^{\infty} h(t)g^*\left(t-(nT - t_l - m\tau)\right)dt \\
\label{SingleCh:eq_phi_inf_sum} &=& \sum_{m\in\mathbb{Z}} \sum_{l=1}^L a_l \varphi(nT - t_l - m\tau),
\end{eqnarray}
where we defined
\begin{equation}
  \varphi(\vartheta) = \langle g(t-\vartheta),h(t) \rangle.
\end{equation}
Since $g(t)$ in \eqref{SingleCh:eq_g_t} vanishes for all $|t|>\tau/2$ and $h(t)$ satisfies \eqref{SingleCh:eq_h_t_support_tau_div2}, the
support of $\varphi(t)$ is $(R+\tau)$, i.e.,
\begin{equation}\label{SingleCh:eq_phi_support}
  \varphi(t) = 0 \quad \textrm{for all } |t| \geq (R + \tau)/2.
\end{equation}
Using this property, the summation in \eqref{SingleCh:eq_phi_inf_sum} will be over nonzero values for indices $m$ satisfying
\begin{equation}\label{SingleCh:eq_phi_condition}
  |nT-t_l-m\tau|< (R+\tau)/2.
\end{equation}

Sampling within the window $[0,\tau)$, i.e., $nT\in[0,\tau)$, and noting that the time-delays lie in the interval
$t_l\in[0,\tau),\, l=1\ldots L$, \eqref{SingleCh:eq_phi_condition} implies that
\begin{equation}
  (R+\tau)/2 > |nT-t_l-m\tau| \geq |m|\tau - |nT-t_l| > (|m|-1)\tau.
\end{equation}
Here we used the triangle inequality and the fact that $|nT-t_l| < \tau$ in our setting. Therefore,
\begin{equation}\label{SingleCh:eq_m_condition}
  |m| < \frac{R/\tau+3}{2} \Rightarrow |m| \leq \left\lceil \frac{R/\tau+3}{2} \right\rceil - 1 \defn r,
\end{equation}
i.e., the elements of the sum in \eqref{SingleCh:eq_phi_inf_sum} vanish for all $m$ but the values in \eqref{SingleCh:eq_m_condition}.
Consequently, the infinite sum in \eqref{SingleCh:eq_phi_inf_sum} reduces to a finite sum over $m\leq |r|$ so that \eqref{SingleCh:eq_periodic_case_samples_phi_introduced} becomes
\begin{eqnarray}
\nonumber
c[n] &=& \sum_{m=-r}^r \sum_{l=1}^L a_l \varphi(nT - t_l - m\tau) \\
\nonumber
&=& \sum_{m=-r}^r \sum_{l=1}^L a_l \int_{-\infty}^{\infty} h(t-t_l)g^*(t-nT+m\tau)dt \\
&=& \left\langle \sum_{m=-r}^r g(t-nT+m\tau),\sum_{l=1}^L a_l h(t-t_l) \right\rangle,
\end{eqnarray}
where in the last equality we used the linearity of the inner product. Defining a function which consists of $(2r+1)$
periods of $g(t)$:
\begin{equation}\label{SingleCh:eq_g_r}
  g_r(t) = \sum_{m=-r}^r g(t+m\tau),
\end{equation}
we conclude that
\begin{eqnarray}\label{SingleCh:eq_c_n_finite_final}
c[n] &=& \langle g_r(t-nT),\tilde x(t) \rangle.
\end{eqnarray}
Therefore, the samples $c[n]$ can be obtained by filtering the aperiodic signal $\tilde x(t)$ with the filter $g_r^*(-t)$
prior to sampling. This filter has compact support equal to
$(2r+1)\tau$. Since the finite setting samples \eqref{SingleCh:eq_c_n_finite_final} are identical to those of the periodic case \eqref{SingleCh:eq_phi_inf_sum}, recovery of the delays and amplitudes is performed exactly the same as in the periodic setting.

We summarize this result in the following theorem.
\begin{theorem}
  Consider the finite stream of pulses given by:
  \begin{equation*}
    \tilde x(t) = \sum_{l=1}^L a_l h(t - t_l),\, t_l \in [0,\tau),\, a_l \in \mathbb{R},
  \end{equation*}
  where $h(t)$ has finite support $R$. Choose a set $\mathcal{K}$ of consecutive indices for which $H(2\pi k/\tau) \neq 0,\, \forall k\in\mathcal{K}$. Then, $N$ samples given
  by:
  \begin{equation*}
    c[n] = \langle g_r(t-nT),\tilde x(t) \rangle,\quad n=0\ldots N-1,\, nT\in[0,\tau),
  \end{equation*}
  where $r$ is defined in \eqref{SingleCh:eq_m_condition}, and $g_r(t)$ is compactly supported and defined by \eqref{SingleCh:eq_g_r} (based on the filter $g(t)$ in \eqref{SingleCh:eq_G_omega}), uniquely determine the signal $\tilde x(t)$
  as long as $N \geq |\mathcal{K}| \geq 2L$.
\end{theorem}

If, for example, the support $R$ of $h(t)$ satisfies $R \leq \tau$ then we obtain from \eqref{SingleCh:eq_m_condition} that $r=1$.
Therefore, the filter in this case would consist of $3$ periods of $g(t)$:
\begin{equation}\label{SingleCh:eq_g_3p}
  g_{3p}(t) \defn g_r(t)\big|_{r=1} = g(t-\tau) + g(t) + g(t+\tau).
\end{equation}
Practical implementation of the filter may be carried out using delay-lines. The relation of this scheme to previous approaches will be investigated in Section~\ref{SingleCh:sec_related_work}.

\subsection{Simulations}\label{SingleCh:subsec_Simulations_Finite}
\subsubsection{Demonstration of the Sampling Scheme}
The input signal $\tilde x(t)$ consists of $L=5$ delayed and weighted versions of the pulse $h(t)=\delta(t)$. The delays
and weights were chosen randomly. We choose $\mathcal{K}=\{-L,\ldots,L\}$, so that $M=|\mathcal{K}|=11$. Since the support of $h(t)$ satisfies $R\leq \tau$ the parameter $r$ in
\eqref{SingleCh:eq_m_condition} equals $1$, and therefore we filter $\tilde x(t)$ with $g_{3p}(t)$ defined in \eqref{SingleCh:eq_g_3p}. The
coefficients $b_k,\, k\in\mathcal{K}$ were all set to one. The output of the filter is sampled uniformly $N$ times, with
sampling period $T = \tau/N$, where $N=M=11$. \RonenComment{
\begin{figure}[h]
\centering
\includegraphics{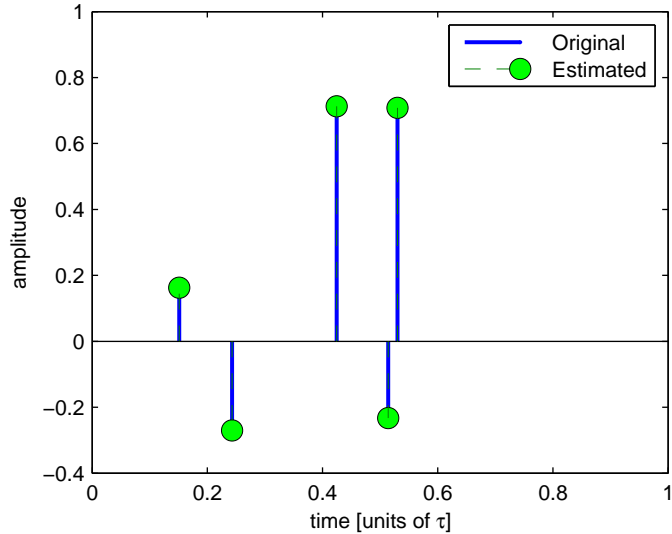}
\caption{Application of the filter $g_{3p}(t)$ on a finite stream of $L=5$ diracs.} \label{SingleCh:fig_aperidic_demonstration}
\end{figure}}
Perfect reconstruction is achieved as can be seen in Fig.\RonenComment{~\ref{SingleCh:fig_aperidic_demonstration}}. The
estimation is exact to numerical precision.

\subsubsection{High Order Problems}
\RonenComment{
\begin{figure}
\centering
\includegraphics{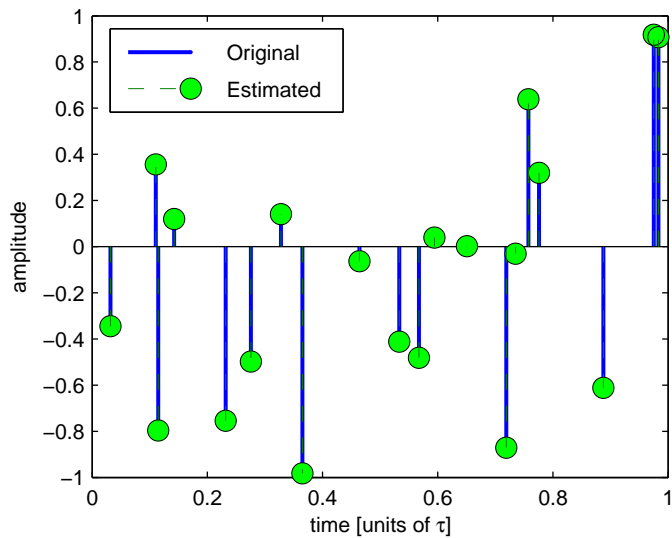}
\caption{High order problems: application of the filter $g_{3p}(t)$ on a finite stream of $L=20$ diracs.}
\label{SingleCh:fig_aperidic_demonstration_K20}
\end{figure}}
The same simulation was carried out with $L=20$ diracs. The results are shown in
Fig.\RonenComment{~\ref{SingleCh:fig_aperidic_demonstration_K20}}. Here again, the reconstruction is perfect even for large $L$.

\subsubsection{Noisy Case}
\RonenComment{
\begin{figure*}
\centering \mbox {
\subfigure[$L=2$]{\includegraphics[scale=0.8]{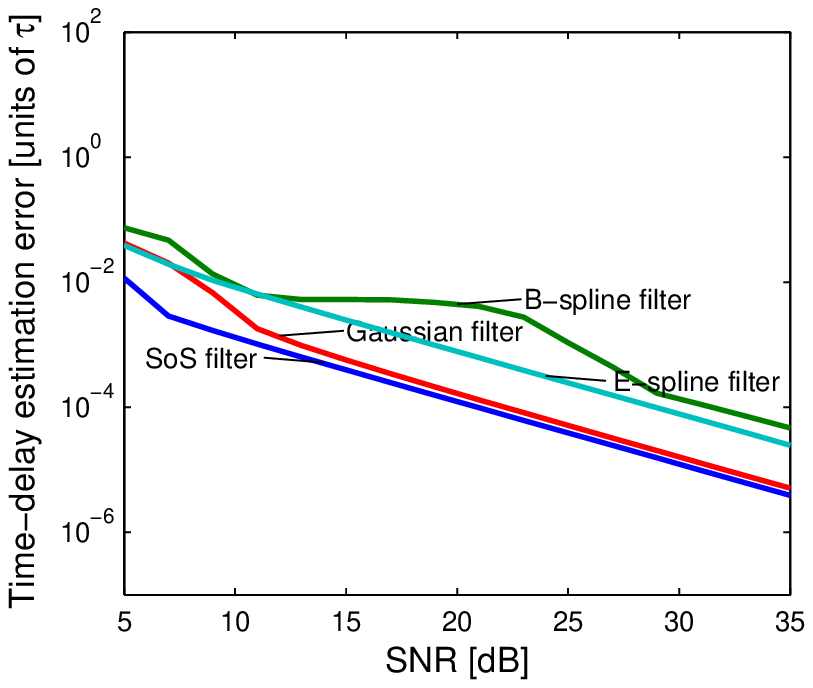}}
\subfigure[$L=3$]{\includegraphics[scale=0.8]{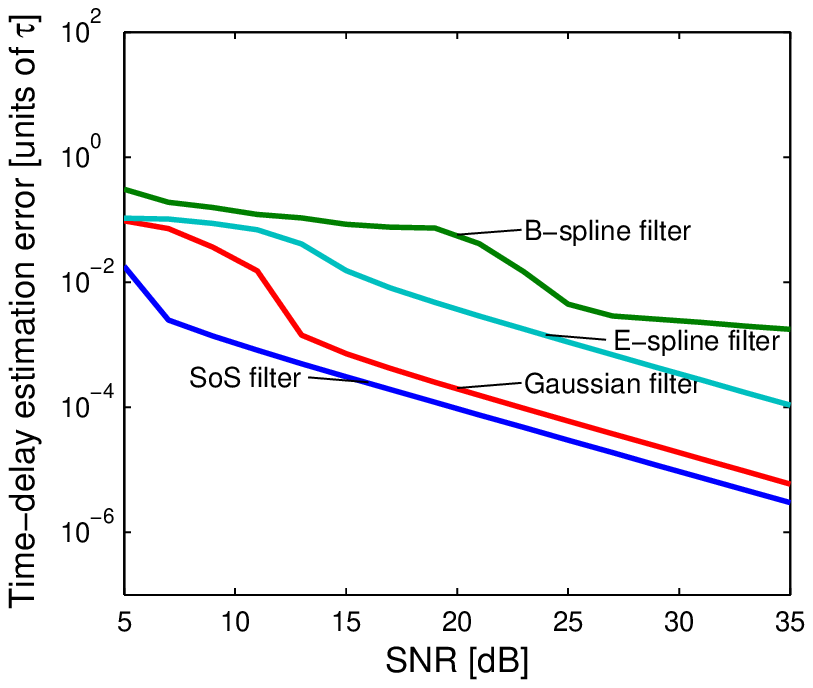}}
} \\
\centering \mbox { \subfigure[$L=5$]{\includegraphics[scale=0.8]{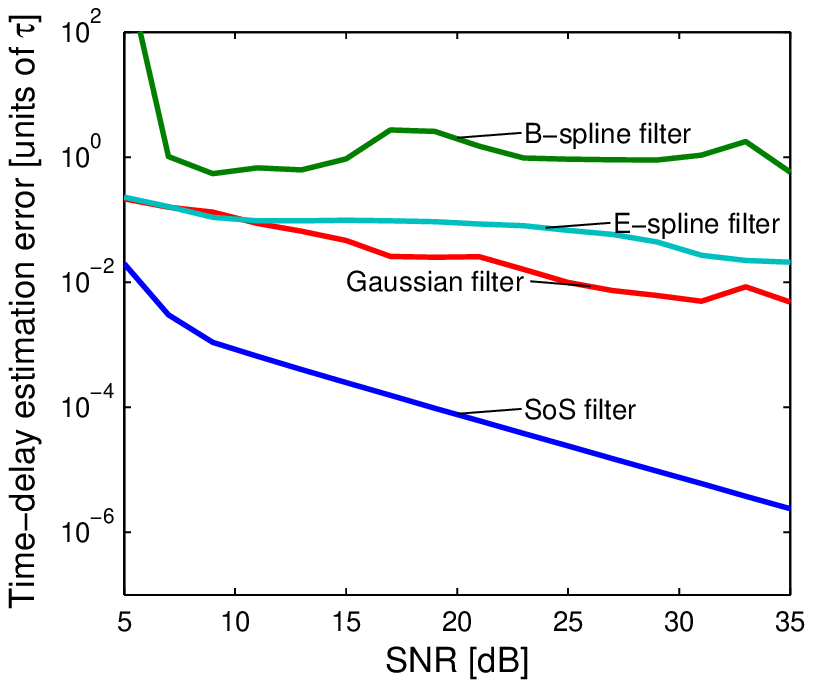}}
\subfigure[$L=20$]{\includegraphics[scale=0.8]{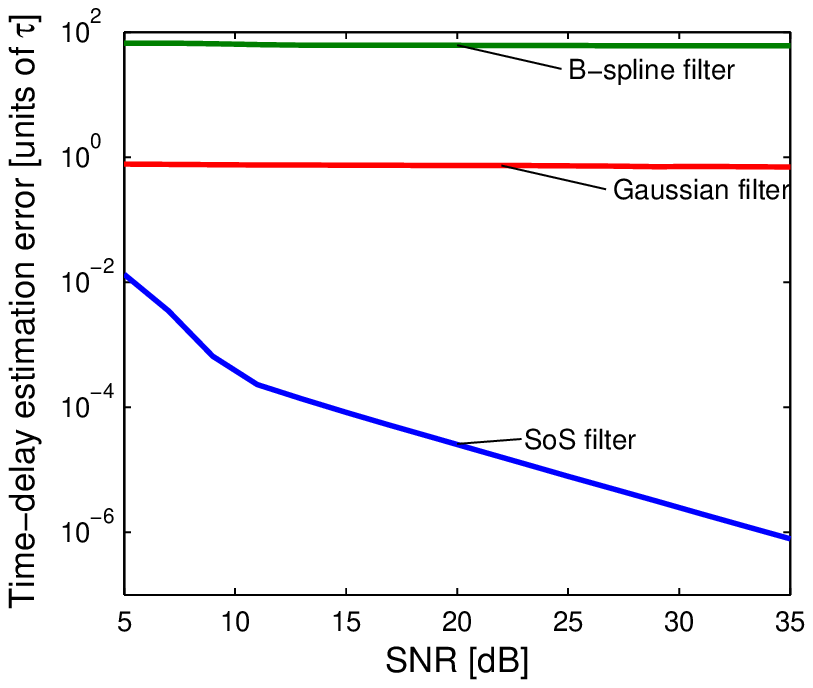}} } \caption{Performance in the presence of noise:
finite stream case. Our method vs. B-spline, E-spline \cite{DragottiStrangFix2007} and Gaussian \cite{Vetterli2002_Basic} sampling
kernels. (a) $L=2$ dirac pulses are present, (b) $L=3$ pulses, (c) high value of $L=5$ pulses, and (d) the performance for
a very high value of $L=20$ (without E-spline simulation, due to computational complexity of calculating the time-domain expression for high values of $L$).} \label{SingleCh:fig_Finite_Stream_Noisy}
\end{figure*}
}
We now consider the performance of our method in the presence of noise. In addition, we compare our performance to the B-spline and E-spline methods proposed in \cite{DragottiStrangFix2007}, and to the Gaussian sampling kernel \cite{Vetterli2002_Basic}. We examine 4 scenarios, in which the signal consists of $L=2,3,5,20$ diracs\footnote{Due to computational complexity of calculating the time-domain expression for high order E-splines, the functions were simulated up to order 9, which allows for $L=5$ pulses.}. In our setup, the
time-delays are equally distributed in the window $[0,\tau)$, with $\tau=1$, and remain constant throughout the experiments. All amplitudes are set to one.

The index set of the SoS filter is $\mathcal{K}=\{-L,\ldots,L\}$. Both B-splines and E-splines are taken of order $2L-1$, and for E-splines we use purely imaginary exponents, equally distributed around the complex unit circle. The sampling period for all methods is $T=\tau/N$.

The method of noise corruption is the same as in Section~\ref{SingleCh:subsubsec_periodic_noisy}. In order to maintain the same SNR conditions throughout all methods, the noise level is chosen with respect to the resulting sequence of samples. In other words, $\sigma_n$ in \eqref{SingleCh:eq_SNR_definition} is method-dependent, and is determined by the desired SNR and the samples of the specific technique. Hard thresholding was implemented in order to improve
the spline methods, as suggested by the authors in \cite{DragottiStrangFix2007}. The threshold was chosen to be
$3\sigma_n$, where $\sigma_n$ is the standard deviation of the AWGN. For the Gaussian sampling kernel the parameter
$\sigma$ was optimized and took on the value of $\sigma=0.25,0.28,0.32,0.9$, respectively.

The results are given in
Fig.\RonenComment{~\ref{SingleCh:fig_Finite_Stream_Noisy}}. For $L=2$ all methods are stable, where E-splines exhibit better performance than B-splines, and Gaussian and SoS approaches demonstrate the lowest errors. As the value of $L$ grows, the advantage of the SoS filter becomes more prominent, where for $L\geq 5$, the performance of Gaussian and both spline methods deteriorate and have errors approaching the order of $\tau$. In contrast, the SoS filter retains its performance nearly unchanged even up to $L=20$, where the B-spline and Gaussian methods are unstable. The improved
version of the Gaussian approach presented in
\cite{Vetterli_FRI_InThePresenceOfNoise} would not perform better
in this high order case, since it fails for $L>9$, as
noted by the authors. A comparison of our approach to previous methods will be detailed in Section~\ref{SingleCh:sec_related_work}.

\section{Infinite Stream of Pulses}\label{SingleCh:sec_infiniteFRI}
We now consider the case of an infinite stream of pulses
\begin{equation}\label{SingleCh:eq_sig_model_infinite}
  z(t) = \sum_{l\in\mathbb{Z}} a_l h(t - t_l),\quad t_l, a_l \in \mathbb{R}.
\end{equation}
We assume that the infinite signal has a bursty character, i.e., the signal has two distinct phases: a) bursts of maximal duration $\tau$ containing at most $L$ pulses, and b) quiet phases between bursts. For the sake of clarity we begin with the case $h(t)=\delta(t)$. For this choice the filter $g_r^*(-t)$ in \eqref{SingleCh:eq_g_r}
reduces to $g_{3p}^*(-t)$ of \eqref{SingleCh:eq_g_3p}.

Since the filter $g_{3p}^*(-t)$ has compact support $3\tau$ we
are assured that the current burst cannot influence samples taken $3\tau/2$ seconds before or after it. In the
finite case we have confined ourselves to sampling within the interval $[0,\tau)$. Similarly, here, we assume that the samples are taken during the burst duration. Therefore, if the minimal spacing between any two consecutive bursts is
$3\tau/2$, then we are guaranteed that each sample taken during the burst is influenced by one burst only, as depicted in Fig.\RonenComment{~\ref{SingleCh:fig_bursty_signal}}. Consequently, the infinite problem can be reduced to a sequential
solution of local distinct finite order problems, as in Section \ref{SingleCh:sec_aperiodicFRI}. Here the compact support of our
filter comes into play, allowing us to apply local reconstruction methods. \RonenComment{
\begin{figure}[h]
\centering
\includegraphics[scale=0.8]{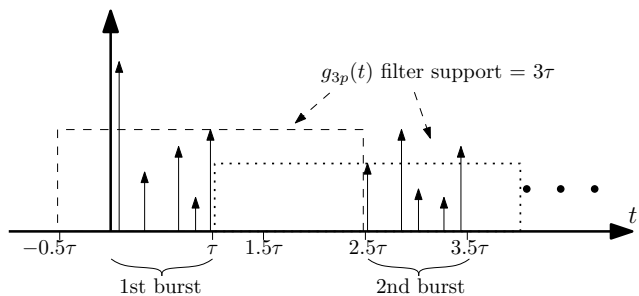}
\caption{Bursty signal $z(t)$. Spacing of $3\tau/2$ between bursts ensures that the influence of the current burst ends
before taking the samples of the next burst. This is due to the finite support, $3\tau$ of the sampling kernel
$g_{3p}^*(-t)$.} \label{SingleCh:fig_bursty_signal}
\end{figure}}

In the above argument we assume we know the locations of the bursts, since we must acquire samples from within the burst duration. Samples outside the burst duration are contaminated by energy from adjacent bursts. Nonetheless, knowledge of burst locations is available in many applications such as synchronized communication
where the receiver knows when to expect the bursts, or in radar or imaging scenarios where the transmitter is itself the
receiver.

We now state this result in a theorem.
\begin{theorem}
  Consider a signal $z(t)$ which is a stream of bursts consisting of delayed and weighted diracs. The maximal burst duration
  is $\tau$, and the maximal number of pulses within each burst is $L$. Then, the samples given by
  \begin{equation*}
    c[n] = \langle g_{3p}(t-nT),z(t) \rangle,\quad n\in\mathbb{Z}
  \end{equation*}
  where $g_{3p}(t)$ is defined by \eqref{SingleCh:eq_g_3p},
  are a sufficient characterization of $z(t)$ as long as the spacing between two adjacent bursts is greater
  than $3\tau/2$, and the burst locations are known.
\end{theorem}
Extending this result to a general pulse $h(t)$ is quite straightforward, as long as $h(t)$ is compactly supported with
support $R$, and we filter with $g_r^*(-t)$ as defined in \eqref{SingleCh:eq_g_r} with the appropriate $r$ from \eqref{SingleCh:eq_m_condition}. If we can choose a set $\mathcal{K}$ of consecutive indices for which
$H(2\pi k/\tau) \neq 0,\, \forall k\in\mathcal{K}$ and we are guaranteed
that the minimal spacing between two adjacent bursts is greater than $\left( (2r+1)\tau + R\right)/2$, then the above theorem holds.

\section{Related Work}\label{SingleCh:sec_related_work}
In this section we explore the relationship between our approach and previously developed solutions \cite{Vetterli2002_Basic,VetterliMagazine2008,DragottiStrangFix2007,UnserFRI2008}.
\subsection{Periodic Case}
The work in \cite{Vetterli2002_Basic} was the first to address
efficient sampling of pulse streams, e.g., diracs. Their approach
for solving the periodic case was ideal lowpass filtering,
followed by uniform sampling, which allowed to obtain the Fourier
series coefficients of the signal. These coefficients are then
processed by the annihilating filter to obtain the unknown
time-delays and amplitudes. In Section~\ref{SingleCh:sec_periodicFRI}, we
derived a general condition on the sampling kernel
\eqref{SingleCh:eq_S_omega}, under which recovery is guaranteed. The
lowpass filter of \cite{Vetterli2002_Basic} is a special
case of this result. The noise robustness of both the lowpass
approach and our more general method is high as long as the pulses
are well separated, since reconstruction from Fourier series
coefficients is stable in this case. Both approaches achieve the minimal number of samples.

The lowpass filter is bandlimited and consequently has infinite
time-support. Therefore, this sampling scheme is unsuitable for
finite and infinite streams of pulses. The SoS class introduced in
Section~\ref{SingleCh:sec_periodicFRI} consists of compactly supported
filters which is crucial to enable the extension of our results to
finite and infinite streams of pulses. A comparison between the two methods is shown in Table~\ref{SingleCh:tab_periodic}.
\renewcommand{\tabW}{p{2.2cm}}
\begin{table}[h]
\renewcommand{\arraystretch}{1.3} 
\caption{Periodic case - Comparison with previous work}
\label{SingleCh:tab_periodic}
\centering
\begin{tabular}{p{2.8cm} || p{2.1cm} | p{2.3cm}}
\hline\hline
Feature & Lowpass filter \cite{Vetterli2002_Basic} & Proposed method \\
\hline\hline
Degrees of freedom & \multicolumn{2}{|c}{$2L$} \\
\hline
No. of samples & $2L+1$ & $2L$ \\
\hline
Time-support & Infinite & $\tau$, finite support allows extension to finite \& infinite cases \\
\hline
Noise Robustness & High & High \\
\hline
Analog implementation & Approximate lowpass filter & Approximate finite support filter \\
\hline\hline
\end{tabular}
\end{table}
\subsection{Finite Pulse Stream}
The authors of \cite{Vetterli2002_Basic} proposed a Gaussian sampling kernel for sampling finite streams of Diracs. The Gaussian method is numerically unstable, as mentioned in \cite{Vetterli_FRI_InThePresenceOfNoise}, since the samples are multiplied by a rapidly diverging or decaying exponent. Therefore, this approach is unsuitable for $L\geq 6$. Modifications proposed in \cite{Vetterli_FRI_InThePresenceOfNoise} exhibit better performance and stability. However, these methods require substantial oversampling, and still exhibit instability for $L>9$.

In \cite{DragottiStrangFix2007} the family of polynomial reproducing kernels was introduced as sampling filters for the model \eqref{SingleCh:eq_sig_model_aperiodic}. B-splines were proposed as a specific example. The B-spline sampling filter enables obtaining moments of the signal, rather than Fourier coefficients. The moments are then processed with the same annihilating filter used in previous methods. However, as mentioned by the authors, this approach is unstable for high values of $L$. This is due to the fact that in contrast to the
estimation of Fourier coefficients, estimating high order moments
is unstable, since unstable weighting of the
samples is carried out during the process.

Another general family introduced in \cite{DragottiStrangFix2007} for the finite model is the class of exponential reproducing kernels. As a specific case, the authors propose E-spline sampling kernels. The CTFT of an E-spline of order $N+1$ is described by
\begin{equation}\label{SingleCh:eq_Espline_CTFT}
  \hat\beta_{\boldsymbol\alpha}(\omega) = \prod_{n=0}^{N} \frac{1-e^{\alpha_n-j\omega}}{j\omega-\alpha_n},
\end{equation}
where $\boldsymbol\alpha = (\alpha_0,\alpha_1,\ldots,\alpha_N)$ are free parameters. In order to use E-splines as sampling kernels for pulse streams, the authors propose a specific structure on the $\alpha$'s, $\alpha_n = \alpha_0 + n\lambda$. Choosing exponents having a non-vanishing real part results in unstable weighting, as in the B-spline case. However, choosing the special case of pure imaginary exponents in the E-splines, already suggested by the authors, results in a reconstruction method based on Fourier coefficients, which demonstrates an interesting relation to our method. The Fourier coefficients are obtained by applying a matrix consisting of the exponent spanning coefficients $\{c_{m,n}\}$, (see \cite{DragottiStrangFix2007}), instead of our Vandermonde matrix relation \eqref{SingleCh:eq_obtaining_x_vector_from_samples}. With this specific choice of parameters the E-spline function satisfies \eqref{SingleCh:eq_S_omega}.

Interestingly, with a proper choice of spanning coefficients, it can be shown that the SoS class can reproduce exponentials with frequencies $\{2\pi k/\tau\}_{k\in\mathcal{K}}$, and therefore satisfies the general exponential reproduction property of \cite{DragottiStrangFix2007}. However, the SoS filter proposes a new sampling scheme which has substantial advantages over existing methods including E-splines. The first advantage is in the presence of noise, where both methods have the following structure:
\begin{equation}\label{SingleCh:eq_Transformation_CondNumber}
  \mathbf{y} = \mathbf{A}\mathbf{x} + \mathbf{w},
\end{equation}
where $\mathbf{w}$ is the noise vector. While the Fourier coefficients vector $\mathbf{x}$ is common to both approaches, the linear transformation $\mathbf{A}$ is method dependent, and therefore the sample vector $\mathbf{y}$ is different. In our approach with $g(t)$ of \eqref{SingleCh:eq_Dirichlet_filter}, $\mathbf{A}$ is the DFT matrix, which for any order $L$ has a condition number of $1$. However, in the case of E-splines the transformation matrix $\mathbf{A}$ consists of the E-spline exponential spanning coefficients, which has a much higher condition number, e.g., above 100 for $L=5$. Consequently, some Fourier coefficients will have much higher values of noise than others. This scenario of high variance between noise levels of the samples is known to deteriorate the performance of spectral analysis methods \cite{Stoica1997}, the annihilating filter being one of them. This explains our simulations which show that the SoS filter outperforms the E-spline approach in the presence of noise.

When the E-spline coefficients $\alpha$ are pure imaginary, it can be easily shown that \eqref{SingleCh:eq_Espline_CTFT} becomes a multiplication of shifted sincs. This is in contrast to the SoS filter which consists of a sum of sincs in the frequency domain. Since multiplication in the frequency domain translates to convolution in the time domain, it is clear that the support of the E-spline grows with its order, and in turn with the order of the problem $L$. In contrast, the support of the SoS filter remains unchanged. This observation becomes important when examining the infinite case. The constraint on the signal in \cite{DragottiStrangFix2007} is that no more than $L$ pulses be in any interval of length $LPT$, $P$ being the support of the filter, and $T$ the sampling period. Since $P$ grows linearly with $L$, the constraint cast on the infinite stream becomes more stringent, quadratically with $L$. On the other hand, the constraint on the infinite stream using the SoS filter is independent of $L$.

We showed in simulations that typically for $L\geq 5$ the estimation errors, using both B-spline and E-spline sampling kernels, become very large. In contrast, our approach leads to stable reconstruction even for very high values of $L$, e.g., $L=100$. In addition, even for low values of $L$ we showed in simulations that although the E-spline method has improved performance over B-splines, the SoS reconstruction method outperforms both spline approaches. A comparison is described in Table~\ref{SingleCh:tab_finite}.
\renewcommand{\tabW}{p{1.5cm}}
\begin{table}[h]
\renewcommand{\arraystretch}{1.3}
\caption{Finite case - comparison}
\label{SingleCh:tab_finite}
\centering
\begin{tabular}{p{2.3cm} || \tabW | \tabW | \tabW}
\hline\hline
Feature & Gaussian filter \cite{Vetterli2002_Basic} & Spline Filter \cite{DragottiStrangFix2007} & Proposed method \\
\hline\hline
Degrees of freedom & \multicolumn{3}{|c}{$2L$} \\
\hline
No. of samples & \multicolumn{3}{|c}{$2L$} \\
\hline
Time-support & Infinite & Finite & Finite \\
\hline
Stability & Unstable for $L\geq 6$ & Unstable for $L \geq 5$ & Stable even for $L=100$ \\
\hline
Noise Robustness & Low & Low & High \\
\hline\hline
\end{tabular}
\end{table}

\subsection{Infinite Streams}
The work in \cite{DragottiStrangFix2007} addressed the infinite stream case, with $h(t) = \delta(t)$.
They proposed filtering the signal with a polynomial reproducing sampling kernel prior to sampling. If the signal has at most $L$ diracs within any interval of duration $LPT$, where $P$ denotes the support of the sampling filter and $T$ the sampling period, then the samples are a sufficient characterization of the signal. This condition allows to divide the infinite stream into a sequence of finite case problems. In our approach the quiet phases of $1.5\tau$ between the bursts of length $\tau$ enable the reduction to the finite case.

Since the infinite solution is based on the finite one, our method is advantageous in terms of stability in high order problems and noise robustness. However, we do have an additional requirement of quiet phases between the bursts.

Regarding the sampling rate, the number of degrees of freedom of the signal per unit time, also known as the rate
of innovation, is $\rho=2L/2.5\tau$, which is the critical sampling rate. Our sampling rate is $2L/\tau$ and therefore we oversample by a factor of $2.5$. In the same scenario, the method in
\cite{DragottiStrangFix2007} would require a sampling rate of $LP/2.5\tau$, i.e., oversampling by
a factor of $P/2$. Properties of polynomial reproducing kernels imply that $P\geq 2L$, therefore for any $L\geq 3$, our method exhibits more efficient sampling. A table comparing the various features is shown in Table~\ref{SingleCh:tab_infinite}.

Recent work \cite{UnserFRI2008} presented a low complexity method for reconstructing streams of pulses (both infinite
and finite cases) consisting of diracs. However the basic assumption of this method is that there is at most one dirac per
sampling period. This means we must have prior knowledge about a lower limit on the spacing between two consecutive
deltas, in order to guarantee correct reconstruction. In some cases such a limit may not exist; even if it
does it will usually force us to sample at a much higher rate than the critical one.
\renewcommand{\tabW}{p{2.4cm}}
\begin{table}[h]
\renewcommand{\arraystretch}{1.3}
\caption{Infinite case - Comparison}
\label{SingleCh:tab_infinite}
\centering
\begin{tabular}{\tabW || \tabW | \tabW}
\hline\hline
Feature & Spline filter \cite{DragottiStrangFix2007} & Proposed method \\
\hline\hline
Signal model & No more than $L$ pulses in any interval of $LPT \textrm{ sec}$ & Bursty character: burst - $\tau$, quiet phase $1.5\tau$ \\
\hline
Rate of innovation & \multicolumn{2}{|c}{$\rho \triangleq 2L/2.5\tau$} \\
\hline
Sampling rate & $P\cdot \rho/2$ & $2.5\rho$ \\
\hline
For $L \geq 3 \quad \Rightarrow P/2 \geq 3$ & \multicolumn{2}{|c}{Proposed sampling scheme is more efficient} \\
\hline
Noise Robustness & Low & High \\
\hline
Stability & Unstable for $L \geq 5$ & Stable for $L=100$ \\
\hline\hline
\end{tabular}
\end{table}

\section{Application - Ultrasound Imaging}\label{SingleCh:sec_ultrasound}
An interesting application of our framework is ultrasound imaging.
In ultrasonic imaging an acoustic pulse is transmitted into the
scanned tissue. The pulse is reflected due to changes in acoustic
impedance which occur, for example, at the boundaries between two
different tissues. At the receiver, the echoes are recorded, where
the time-of-arrival and power of the echo indicate the scatterer's
location and strength, respectively. Accurate estimation of tissue
boundaries and scatterer locations allows for reliable detection
of certain illnesses, and is therefore of major clinical
importance. The location of the boundaries is often more
important than the power of the reflection. This stream of
pulses is finite since the pulse energy decays within the tissue.
We now demonstrate our method on real 1-dimensional (1D)
ultrasound data.

The multiple echo signal which is recorded at the receiver can be modeled as a finite stream of pulses, as in \eqref{SingleCh:eq_sig_model_aperiodic}. The unknown time-delays correspond to the locations of the various scatterers, whereas the amplitudes correspond to their reflection coefficients. The pulse shape in this case is a Gaussian defined in \eqref{SingleCh:eq_gaussian_pulse}, due the physical characteristics of the electro-acoustic transducer (mechanical damping). We assume the received pulse-shape is known, either by assuming it is unchanged through propagation, through physically modeling ultrasonic wave propagation, or by prior estimation of received pulse. Full investigation of mismatch in the pulse shape is left for future research.

In our setting, a phantom consisting of uniformly spaced pins, mimicking point scatterers, was scanned by GE Healthcare's Vivid-i portable ultrasound imaging system \cite{senior2006portable,mondillo2006hand}, using a 3S-RS probe. We use the data recorded by a single element in the probe, which is modeled as a 1D stream of pulses. The center frequency of the probe is $f_c = 1.7021\textrm{ MHz}$, The width of the transmitted Gaussian
pulse in this case is $\sigma = 3\cdot 10^{-7}\textrm{ sec}$, and the depth of imaging is $R_{\textrm{max}} = 0.16\textrm{ m}$ corresponding to a time window of\footnote{The speed of sound within the tissue is $1550 \textrm{ m/sec}$.} $\tau = 2.08\cdot 10^{-4}\textrm{ sec}$.

In this experiment all filtering and sampling operations are carried out digitally in simulation. The analog filter required by the sampling scheme is replaced by a lengthy Finite Impulse Response (FIR) filter. Since the sampling frequency of the element in the system is $f_s = 20\textrm{ MHz}$, which is more than $5$ times higher than the Nyquist rate, the recorded data represents the continuous signal reliably. Consequently, digital filtering of the high-rate sampled data vector ($4160$ samples) followed by proper decimation mimics the original analog sampling scheme with high accuracy. The recorded signal is
depicted in Fig.\RonenComment{~\ref{SingleCh:fig_RealDataRecorded}}. The band-pass ultrasonic signal is demodulated to base-band,
i.e., envelope-detection is performed, before inserted into the process.
\begin{figure}[h]
\centering
\includegraphics{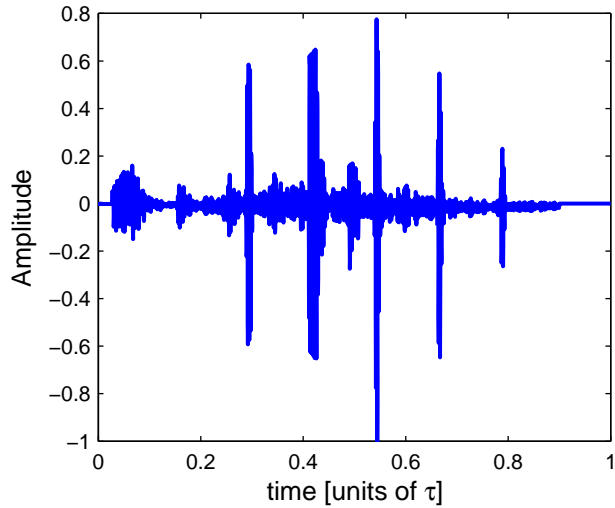}
\caption{Recorded ultrasound imaging signal. The data was acquired by GE healthcare's Vivid-i ultrasound imaging system.}
\label{SingleCh:fig_RealDataRecorded}
\end{figure}

We carried out our sampling and reconstruction scheme on the aforementioned data. We set $L=4$, looking for the strongest
4 echoes. Since the data is corrupted by strong noise we over-sampled the signal, obtaining twice the minimal number of samples. In addition, hard-thresholding of the samples was implemented, where we set the threshold to 10 percent of the maximal value. We obtained $N=17$ samples by decimating the output of the lengthy FIR digital filter imitating $g^*_{3p}(-t)$ from \eqref{SingleCh:eq_g_3p}, where the
coefficients $\{b_k\}$ were all set to one. In Fig.\RonenComment{~\ref{SingleCh:fig_RealDataReconstruction}}a the reconstructed
signal is depicted vs. the full demodulated signal using all $4160$ samples. Clearly, the time-delays were estimated with high precision. The
amplitudes were estimated as well, however the amplitude of the second pulse has a large error. This is probably due to
the large values of noise present in its vicinity. However, as mentioned earlier, the exact locations of the scatterers is often more important than the accurate reflection coefficients. \RonenComment{
\begin{figure*}
\centering \mbox { \subfigure[]{\includegraphics[scale=0.8]{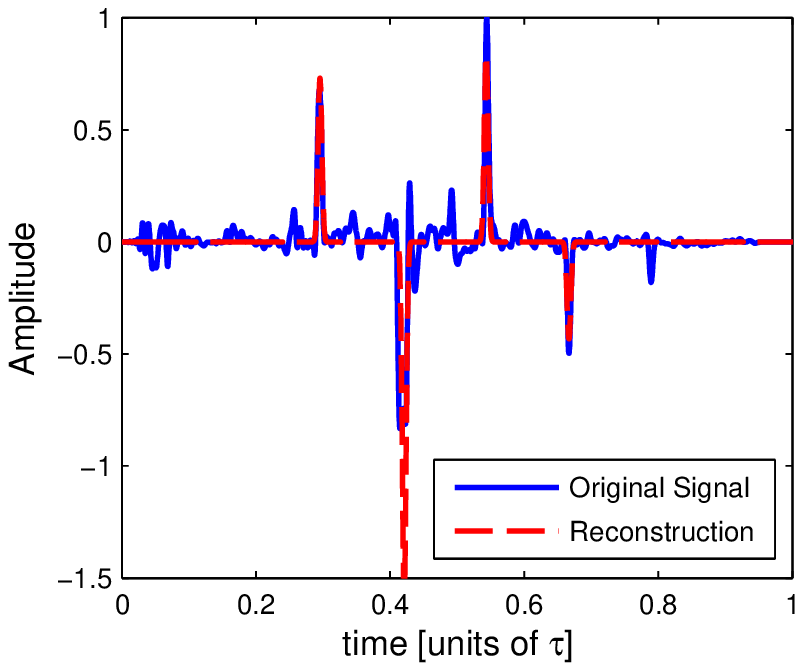}}
\subfigure[]{\includegraphics[scale=0.8]{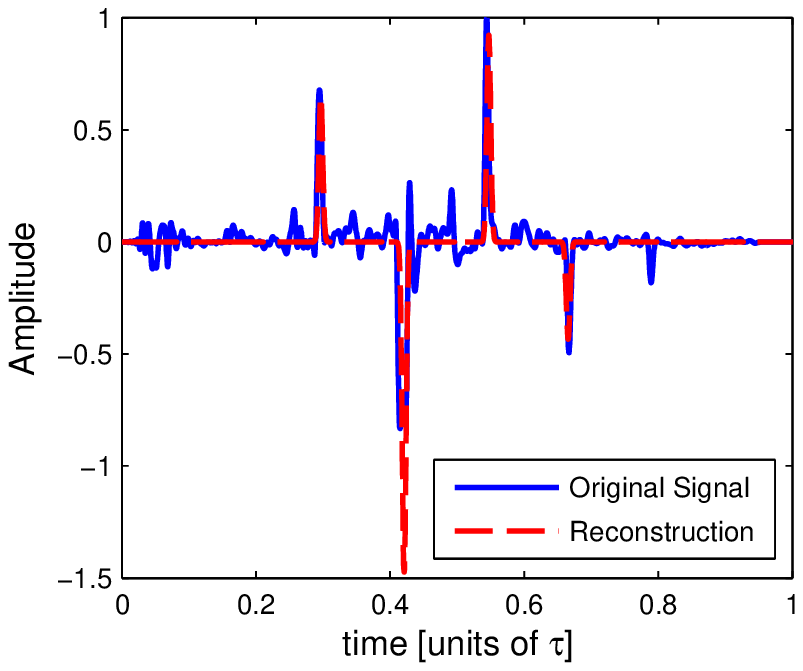}} } \caption{Applying our
$g_{3p}(t)$ filter method on real ultrasound imaging data. Results are shown vs. full demodulated signal which uses all $4160$ samples. Reconstructed signal (a) using $N=17$ samples only and hard-thresholding, and (b) using $N=33$ samples
without thresholding.} \label{SingleCh:fig_RealDataReconstruction}
\end{figure*}
}
We carried out the same experiment only now oversampling by a factor of 4, resulting in $N=33$ samples. Here no
hard-thresholding is required. The results are depicted in Fig.\RonenComment{~\ref{SingleCh:fig_RealDataReconstruction}}b, and are
very similar to our previous results. In both simulations, the estimation error in the pulse location is around $0.1\textrm{ mm}$.

Current ultrasound imaging technology 
operates at the high rate sampled data, e.g., $f_s = 20\textrm{ MHz}$ in our setting. Since there are usually 100 different elements in a single ultrasonic probe each sampled at a very high rate, data throughput becomes very high, and imposes high computational complexity to the system, limiting its capabilities. Therefore, there is a demand for lowering the sampling rate, which in turn will reduce the complexity of reconstruction. Exploiting the parametric point of view, our sampling scheme reduces the sampling rate by 2 orders of magnitude, from 4160 to around 30 samples in our setting, while estimating the locations of the scatterers with high accuracy.

\section{Conclusions}
We presented efficient sampling and reconstruction schemes for
streams of pulses. For the case of a periodic stream of pulses, we
derived a general condition on the sampling kernel which allows a
single-channel uniform sampling scheme. Previous work
\cite{Vetterli2002_Basic} is a special case of this general
result. We then proposed a class of filters, satisfying the
condition, with compact support. Exploiting the compact support of
the filters, we constructed a new sampling scheme for the case of
a finite stream of pulses. Simulations show this method exhibits
better performance than previous techniques \cite{Vetterli2002_Basic,DragottiStrangFix2007}, in terms of stability
in high order problems, and noise robustness. An extension to an
infinite stream of pulses was also presented. The compact support
of the filter allows for local reconstruction, and thus lowers the
complexity of the problem. Finally, we demonstrated the advantage
of our approach in reducing the sampling and processing rate of ultrasound
imaging, by applying our techniques to real ultrasound data.

\section*{Acknowledgements}
The authors would like to thank the anonymous reviewers for their valuable comments.

\appendix
\section*{Proof of Theorem~\ref{SingleCh:theorem_optimal_b_k}}\label{SingleCh:appendix_optimal_b_k}
The MSE of the optimal linear estimator of the vector $\mathbf{x}$ from the measurement vector $\mathbf{y}$ is known to be \cite{kay93}
\begin{equation}\label{SingleCh:appendix_eq_MSE}
  \textrm{MSE} = \textrm{Tr}\left\{ \mathbf{R}_{xx} \right\} -
  \textrm{Tr}\left\{\mathbf{R}_{xy}\mathbf{R}_{yy}^{-1}\mathbf{R}_{yx} \right\}.
\end{equation}
The covariance matrices in our case are
\begin{align}\label{SingleCh:appendix_eq_Rxc}
  \mathbf{R}_{xy} &= \mathbf{R}_{xx}\mathbf{B}^*\mathbf{V}^* \\
\label{SingleCh:appendix_eq_Rcc}
  \mathbf{R}_{yy} &= \mathbf{VB}\mathbf{R}_{xx}\mathbf{B}^*\mathbf{V}^* + \sigma^2\mathbf{I},
\end{align}
where we used \eqref{SingleCh:eq_c_n_matrix_noisy}, and the fact that $\mathbf{R}_{ww} = \sigma^2\mathbf{I}$ since $\mathbf{w}$ is a white Gaussian noise vector. Under our assumptions on $\{t_l\}$ and $\{a_l\}$, denoting $h_k = H(2\pi k/\tau)$, and using \eqref{SingleCh:eq_X_k}
\begin{align}
\nonumber
  \left(\mathbf{R}_{xx}\right)_{k,k'} &= E\left\{X[k]X^*[k']\right\} \\
\nonumber
  &= \frac{1}{\tau^2} h_k h_{k'} \sum_{l=1}^L \sum_{l'=1}^L
    E\left\{a_l a_{l'}^* e^{-j\frac{2\pi}{\tau}(k t_l -k't_{l'})}\right\} \\
\nonumber
  &= \frac{\sigma_a^2}{\tau^2} h_k h_{k'} \sum_{l=1}^L E\left\{e^{-j\frac{2\pi}{\tau}(k-k')t_l}\right\}\\
\nonumber
  &= \frac{\sigma_a^2}{\tau^2}h_k h_{k'}\sum_{l=1}^L \int_0^{\tau} \frac{1}{\tau} e^{-j\frac{2\pi}{\tau}(k-k')t_l} {\rm d}t \\
\label{SingleCh:appendix_eq_Rxx}
  &= \frac{\sigma_a^2}{\tau^2}L |h_k|^2 \delta_{k,k'}.
\end{align}
Denoting by $\mathbf{\tilde H}$ a diagonal matrix with \textit kth element $|\tilde h_k|^2 = |h_k|^2 \sigma_a^2 L / \tau^2$ we have
\begin{equation}\label{SingleCh:appendix_eq_Rxx_final}
  \mathbf{R}_{xx} = \mathbf{\tilde H}.
\end{equation}

Since the first term of \eqref{SingleCh:appendix_eq_MSE} is independent of $\mathbf{B}$, minimizing the MSE with respect to $\mathbf{B}$ is equivalent to maximizing the second term in \eqref{SingleCh:appendix_eq_MSE}. Substituting \eqref{SingleCh:appendix_eq_Rxc},\eqref{SingleCh:appendix_eq_Rcc} and \eqref{SingleCh:appendix_eq_Rxx_final} into this term, the optimal $\mathbf{B}$ is a solution to
\begin{align}
\label{SingleCh:appendix_eq_optimal_B}
  &\max_{\mathbf{B}} \textrm{Tr}\big\{\mathbf{\tilde H}\mathbf{B}^*\mathbf{V}^*
  (\mathbf{VB}\mathbf{\tilde H}\mathbf{B}^*\mathbf{V}^* + \sigma^2\mathbf{I})^{-1}\mathbf{VB\mathbf{\tilde H}}\big\} \\
\nonumber
  &\quad \textrm{s.t. } \textrm{Tr}(\mathbf{B}^*\mathbf{B}) = 1.
\end{align}
Using the matrix inversion formula \cite{golub1996matrix},
\begin{align}
\nonumber
  &(\mathbf{VB}\mathbf{\tilde H}\mathbf{B}^*\mathbf{V}^* + \sigma^2\mathbf{I})^{-1} \\
\label{SingleCh:appendix_eq_2nd_term_take1}
  &= \frac{1}{\sigma^2}\bigg(\mathbf{I} -
  \mathbf{VB}\left( \sigma^2\mathbf{\tilde H}^{-1} + \mathbf{B}^*\mathbf{V}^*\mathbf{VB} \right)^{-1}\mathbf{B}^*\mathbf{V}^*\bigg).
\end{align}
It is easy to verify from the definition of $\mathbf{V}$ in \eqref{SingleCh:eq_c_n_matrix} that
\begin{equation}
  \left(\mathbf{V}^*\mathbf{V}\right)_{ik} =
  \sum_{l=0}^{N-1} e^{j\frac{2\pi}{N}l(k-i)} = N\delta_{k,i}.
\end{equation}
Therefore, the objective in \eqref{SingleCh:appendix_eq_optimal_B} equals
\begin{align}
\nonumber
  &\textrm{Tr}\bigg\{ \frac{N}{\sigma^2} \mathbf{\tilde H}\mathbf{B}^*
  \left(\mathbf{I} - \mathbf{B}\left( \frac{\sigma^2}{N}\mathbf{\tilde H}^{-1} + \mathbf{B}^*\mathbf{B}\right)^{-1}
  \mathbf{B}^*\right)\mathbf{B}\mathbf{\tilde H} \bigg\} \\
\label{SingleCh:appendix_Trace_Result}
  &= \sum_{i=1}^{|\mathcal{K}|} |\tilde h_i|^2 \left(1 - \frac{\sigma^2/N}{|b_i|^2|\tilde h_i|^2+\sigma^2/N} \right)
\end{align}
where we used the fact that $\mathbf{B}$ and $\mathbf{\tilde H}$ are diagonal.

We can now find the optimal $\mathbf{B}$ by maximizing \eqref{SingleCh:appendix_Trace_Result}, which is equivalent to minimizing the negative term:
\begin{equation}\label{SingleCh:appendix_eq_optim_problem}
  \min_{\mathbf{B}} \sum_{i=1}^{|\mathcal{K}|} \frac{|\tilde h_i|^2}{1+|b_i|^2|\tilde h_i|^2 N/\sigma^2} ,\, \textrm{s.t. } \sum_{i=1}^{|\mathcal{K}|} |b_i|^2 = 1.
\end{equation}
Denoting $\beta_i = |b_i|^2$, \eqref{SingleCh:appendix_eq_optim_problem} becomes a convex optimization problem:
\begin{align}\label{SingleCh:appendix_eq_optim_problem_convex}
  \min_{\beta_i} \sum_{i=1}^{|\mathcal{K}|} \frac{|\tilde h_i|^2}{1+\beta_i|\tilde h_i|^2 N/\sigma^2}
\end{align}
subject to
\begin{align}
\label{SingleCh:appendix_eq_constraint_beta_i_positive}
  \beta_i &\geq 0 \\
\label{SingleCh:appendix_eq_constraint_sum_beta_i}
  \sum_{i=1}^{|\mathcal{K}|} \beta_i &= 1.
\end{align}
To solve \eqref{SingleCh:appendix_eq_optim_problem_convex} subject to \eqref{SingleCh:appendix_eq_constraint_beta_i_positive} and \eqref{SingleCh:appendix_eq_constraint_sum_beta_i}, we form the Lagrangian:
\begin{equation}\label{SingleCh:appendix_eq_Lagrangian}
  \mathcal{L} = \sum_{i=1}^{|\mathcal{K}|} \frac{|\tilde h_i|^2}{1+\beta_i|\tilde h_i|^2 N/\sigma^2} + \lambda\left( \sum_{i=1}^{|\mathcal{K}|} \beta_i - 1\right) - \sum_{i=1}^{|\mathcal{K}|} \mu_i \beta_i
\end{equation}
where from the Karush-Kuhn-Tucker (KKT) conditions \cite{bertsekas1999nonlinear}, $\mu_i \geq 0$ and $\mu_i \beta_i=0$. Differentiating \eqref{SingleCh:appendix_eq_Lagrangian} with respect to $\beta_i$ and equating to $0$
\begin{equation}\label{SingleCh:appendix_eq_diff_Lagrangian}
  \frac{|\tilde h_i|^4 N / \sigma^2}{(1+\beta_i|\tilde h_i|^2 N/\sigma^2)^2} + \mu_i = \lambda,
\end{equation}
so that $\lambda > 0$, since $\tilde h_i > 0$ by construction of $\mathbf{H}$ (see Theorem~\ref{SingleCh:prop_periodic_general_codition_filter}). If $\lambda > |\tilde h_i|^4 N/\sigma^2$ then $\mu_i >0$, and therefore, $\beta_i = 0$ from KKT. If $\lambda \leq |\tilde h_i|^4 N/\sigma^2$ then from \eqref{SingleCh:appendix_eq_diff_Lagrangian} $\mu_i=0$ and
\begin{equation}
  \beta_i = \frac{\sigma^2}{N}\left(\sqrt{\frac{N}{\lambda \sigma^2}}-\frac{1}{|\tilde h_i|^2}\right).
\end{equation}
The optimal $\beta_i$ is therefore
\begin{equation}\label{SingleCh:appendix_eq_optimal_beta_i}
  \beta_i = \left\{
  \begin{array}{l l}
    \frac{\sigma^2}{N}\left(\sqrt{\frac{N}{\lambda \sigma^2}}-\frac{1}{|\tilde h_i|^2}\right) & \lambda \leq |\tilde h_i|^4 N/\sigma^2 \\
    0 & \lambda > |\tilde h_i|^4 N/\sigma^2\\
  \end{array} \right.
\end{equation}
where $\lambda>0$ is chosen to satisfy \eqref{SingleCh:appendix_eq_constraint_sum_beta_i}. Note that from \eqref{SingleCh:appendix_eq_optimal_beta_i}, if $\beta_i\neq 0$ and $i<j$, then $\beta_j\neq 0$ as well, since $|\tilde h_i|$ are in an increasing order.
We now show that there is a unique $\lambda$ that satisfies \eqref{SingleCh:appendix_eq_constraint_sum_beta_i}. Define the function
\begin{equation}
  \mathcal{G}(\lambda) = \sum_{i=1}^{|\mathcal{K}|} \beta_i(\lambda) - 1,
\end{equation}
so that $\lambda$ is a root of $\mathcal{G}(\lambda)$. Since the $|\tilde h_i|$'s are in an increasing order, $|\tilde h_{|\mathcal{K}|}| = \max_i |\tilde h_i|$. It is clear from \eqref{SingleCh:appendix_eq_optimal_beta_i} that $\mathcal{G}(\lambda)$ is monotonically decreasing for $0 < \lambda \leq |\tilde h_{|\mathcal{K}|}|^4 N/\sigma^2$. In addition, $\mathcal{G}(\lambda) = -1$ for $\lambda > |\tilde h_{|\mathcal{K}|}|^4 N/\sigma^2$, and $\mathcal{G}(\lambda) > 0$ for $\lambda \rightarrow 0$. Thus, there is a unique $\lambda$ for which \eqref{SingleCh:appendix_eq_constraint_sum_beta_i} is satisfied.

Substituting \eqref{SingleCh:appendix_eq_optimal_beta_i} into \eqref{SingleCh:appendix_eq_constraint_sum_beta_i}, and denoting by $m$ the smallest index for which $\lambda \leq |\tilde h_{m+1}|^4 N/\sigma^2$, we have
\begin{equation}
  \sqrt{\lambda} = \frac{(|\mathcal{K}|-m)\sqrt{N/\sigma^2}}{N/\sigma^2 + \displaystyle\sum_{i=m+1}^{|\mathcal{K}|}1/|\tilde h_i|^2},
\end{equation}
completing the proof of the theorem.

\bibliography{IEEEabrv,Sampling_FRI_Signals_bib}

\end{document}